\documentclass[10pt,a4paper,reqno]{amsart}
\usepackage{acronym}
\usepackage{amsfonts}
\usepackage{amsmath}
\usepackage{amssymb}
\usepackage{amsthm}
\usepackage{bm}
\usepackage{braket}
\usepackage{cite}
\usepackage{color}
\usepackage{geometry}
\usepackage{graphicx}
\usepackage{hyperref}
\usepackage{mathrsfs}
\usepackage{mathtools}
\usepackage{setspace}
\usepackage{stmaryrd}
\usepackage{subcaption}
\usepackage{tikz}

\newgeometry{top=2cm,bottom=2cm,outer=1.5cm,inner=1.5cm}

\newcommand{\bse}{\begin{subequations}}
\newcommand{\ese}{\end{subequations}}

\newtheorem{theorem}{Theorem}

\newtheorem{lemma}[theorem]{Lemma}

\newtheorem{proposition}[theorem]{Proposition}
\numberwithin{equation}{section}

\title[High-order soliton matrix for an extended nonlinear Schr\"{o}dinger equation]{High-order soliton matrix for an extended nonlinear Schr\"{o}dinger equation}
\author{Huijuan Zhou}
\address[HZ]{School of Mathematical Sciences, Shanghai Key Laboratory of Pure Mathematics and Mathematical Practice, and Shanghai Key Laboratory of Trustworthy Computing \\
East China Normal University \\ Shanghai 200241 \\ People's Republic of China}
\author{Yong Chen}
\address[YC]{School of Mathematical Sciences, Shanghai Key Laboratory of Pure Mathematics and Mathematical Practice, and Shanghai Key Laboratory of Trustworthy Computing \\
East China Normal University \\ Shanghai 200241 \\ People's Republic of China}
\address[YC]{College of Mathematics and Systems Science \\ Shandong University of Science and Technology \\ Qingdao 266590 \\ People's Republic of China}
\address[YC]{Department of Physics \\ Zhejiang Normal University \\ Jinhua 321004 \\ People's Republic of China}
\email{ychen@sei.ecnu.edu.cn}

\begin{document}

\begin{abstract}
 The extended nonlinear Schr\"{o}dinger (ENLS) equation with third-order term and fourth-order term which describes the wave propagation in the optical fibers is more accurate than the NLS equation. A study of high-order soliton matrix is presented for an ENLS equation in the framework of the Riemann-Hilbert problem (RHP). Through a standard dressing procedure and the generalized Darboux transformation (gDT), soliton matrix for simple zeros and elementary high-order zeros in the RHP for the ENLS equation are constructed. Then the N-soliton solutions and high-order soliton solutions for the ENLS equation can be determined. Moreover, collision dynamics along with the asymptotic behavior for the two-solitons and long-time asymptotic estimations for the high-order one-soliton are concretely analyzed. For the given spectral parameters, we can control the propagation direction, velocity, width and other physical quantities of solitons by adjusting the free parameters of ENLS equation.
\end{abstract}

\maketitle

\section{Introduction}

The phenomenon of the solitary wave, which was discovered by the famous British scientist Russell in 1834. He thought that this kind of wave should be a stable solution of fluid motion and named it ``solitary waves'', but he had not confirmed the existence of solitary wave in theory. Until 1895, Korteweg and his student de Vries pointed out that such waves could be approximated as long waves with small amplitude, and thus established the KdV equation. The existence of solitary waves is explained theoretically by the KdV equation. In 1955, Fermi, Pasta and Ulam published ``Studies of nonlinear problems'', so that the study of solitary waves is active again. 10 years later, Kruskal and Zabusky, two American mathematicians, studied the whole process of the interaction between two waves of the KdV equation in detail through numerical calculation using advanced computers. The result indicated that the solitary waves have the property of elastic collision, which is similar to the colliding property of particles. Therefore, Kruskal and Zabusky named them ``solitons" \cite{zk1965}. From then on, the research on solitons began to flourish.

The inverse scattering transform (IST) method was discovered by Gardner, Greene, Kruskal and Miura (GGKM) in 1967 as a method to solve the initial value problem that decay sufficiently rapidly at infinity for the KdV equations \cite{ggkm1967}. In 1968, Lax gave Lax pair of KdV equations \cite{lax1968} and pointed out that this IST was general and could solve the initial value problem of multiple equations, and established the general framework of the solving theory of the IST. Zakharov and Shabat promoted IST by using Lax' thought, and gave the solution of the higher-order KdV equation and cubic Schr\"{o}dinger equation in 1972 \cite{zs1972}, which is the first time to give an example to prove the generality of the IST. In 1973, The initial value problem for the sine-Gordon equation is solved by the IST \cite{cy28}. In the same year, Ablowitz, Newell, Manakov and Shabat et al. studied the long-time behavior of KdV equation and NLS equation according to the IST \cite{12, 13, 14}. In 1975-1976, the continuation method (W-E method) of nonlinear PDE, with only two independent variables, was proposed by Wahlpuist and Estabrook \cite{cy291, cy292}, and an important application of which was to obtain Lax pairs of the equation with the help of Lie algebra, providing a necessary condition for solving the equation with IST. However, obtainment of the solution by the W-E method is more complicated. IST is one of the important discoveries in the field of mathematical physics in the 20th century.

 The IST method was originally solved by using the Gel'Fand-Levitan-Marchenko (GLM) integral equation, although GLM equation can be used to obtain the solution of the equation, the solution process is very complex. The RHP, derived by two famous mathematicians Riemann and Hilbert, was first introduced by Riemann in his doctoral thesis in 1851, and then generalized by Hilbert to a more formal form in 1900, and presented at the international congress of mathematicians in Paris.
In 1976, Zakharov and Manakov used precise steps to give a long-time asymptotic formula for the solution of NLS equation that explicitly depends on the initial value \cite{83lm}.  In the 1980s, Jimbo, Miwa et al. applied the IST to the long-time properties of quantum solvable model \cite{41lm}. In essence, these methods in \cite{83lm, 41lm} requiring a priori judgment on the asymptotic form of the solution have implied the ideas of classical RHP. RHP as a more general method than the IST began applied to integrable systems since the 1980s. For example, GLM theory is equivalent to the RHP for second-order spectral problem, while since there is no GLM theory for the high-order spectral problem, inverse scattering problem needs to be transformed into RHP. Most importantly, the exact long-time asymptotic property of the solution can be obtained through RHP.

Inspired by the work of Zakharov and Manakov, Its developed isomonodromy method and converted the long-time behavior of the initial value problems for the NLS equation into a small neighborhood of the local RHP in 1981, providing a set of practical and strict approaches for analyzing the long-time behavior of integrable equations \cite{It1981}. However, this method still cannot get rid of the prior judgment on the asymptotic form of the solution. In 1989, Zhou studied the connection between the Riemann-Hilbert factorization on self-intersecting contours and a class of singular integral equations with a pair of decomposing algebras \cite{k1989}, providing an effective way to treat the IST problem of first-order systems. Deift and Zhou perfected Its's methods, and proposed nonlinear steepest descent method in 1993 \cite{D1993}. The RHP corresponding to the initial value problem of mKdV equation was studied directly by using the nonlinear steepest descent method, and the long-time behavior of the exact solution of mKdV was obtained. In 1997 and 2003, using nonlinear steepest descent method, long-time behavior of the initial value problem of small dispersion KdV equation \cite{23lm} and  the solutions of NLS equation for weighted Sobolev space initial data \cite{25lm} was studied successively by Deift and Zhou. In 2002, Vartanian studied the long-time behavior of solutions of NLS equations with finite dense initial values \cite{75lm}. Tovbis studied the asymptotic properties of the first term of the solution to the semi-classical limit initial value problem of NLS equation in 2006 \cite{73lm}. In 2009, Monvel, Its and Kotlyarov studied the long-time properties of solutions of focusing NLS equation under periodic boundary conditions on a half-line \cite{67lm}. In 2010, Yang used RHP to study the initial value problems of nonlinear integrable system long-time asymptotic behavior of soliton solution \cite{yjks, yjk2003b}. In 2011, Deift and Park studied the solution of the focused NLS equation with Robin boundary conditions at the origin on a half line long-time behavior \cite{22lm}. Fokas published three papers in 2012 on solving the RHP in integrable systems\cite{36lm,60lm,61lm}. Since 2013, Fan's group began to study the RHP \cite{xj2013, xf2015}, and they have studied the long-time asymptotic behavior of Fokas-Lenells equation under zero boundary conditions based on the nonlinear steepest descent method. The initial boundary value problems of Sasa-Satsuma equations and three wave equations with more complex spectral problems were studied by using the Fokas method. In the last five years, many papers have been published on solving the initial boundary value problem and long-time behavior of integrable equations under the RHP framework. For example, according to the Deift-Zhou nonlinear steepest descent method, Biondini studied the long-time asymptotics for the focusing NLS equation with nonzero boundary conditions at infinity and asymptotic stage of modulational instability \cite{b2017}. Miller et al. investigated rogue waves of infinite order and the Painl\'{e}ve-III hierarchy by using the nonlinear steepest descent method \cite{bl2018}. Bilman gave the large-order asymptotics for multiple-pole solitons of the focusing NLS equation \cite{bb22019}.

It is well-known that the classical method of IST shows that the poles of the scattering coefficients (or zeros of the RHP) can produce soliton solutions. The soliton solutions are usually derived by using one of the several well-known techniques, such as the dressing method or the RHP approach. However, in most literatures only soliton solutions from simple poles are considered. It is usually assumed that a multiple-pole solution can be obtained in a straightforward way by coalescing several distinct poles \cite{no1984} which describe multi-soliton solutions.
Indeed, the soliton dressing matrix corresponding to a multi-soliton solution is a rational matrix function which has distinct simple poles, while the coalescing procedure must produce multiple poles. Soliton solutions corresponding to multiple poles,i.e. the high-order solitons, have been investigated in the literatures \cite{yjk2003a, zys12019}. As described in \cite{gj164}, this high-order soliton can be used to describe the weak bound state of a soliton solution and may appear in the study of line-propagation solitons with nearly the same speed and amplitude height. High-order soliton solutions of some equations, such as sine-Gordon, Schr\"{o}dinger, Kadomtsev-Petviashvili I, N-wave system, derivative NLS and Landau-Lifshitz equations have been studied in the following literatures \cite{nwgj9, nwgj2000, nwgj1999, yjk2003b, llm2, llm2015}.

Recently, we have also done some research related to the RHP and high-order soliton, in \cite{ybss2019} we studied the high-order soliton matrix for Sasa-Satsuma equation in the framework of the RHP. It is noted that pairs of zeros are simultaneously tackled in the situation of the higher-order zeros, which is different from other NLS-type equations. Moreover, collision dynamics along with the asymptotic behavior for the two solitons were analyzed, and long-time asymptotic estimations for the higher-order soliton solution concretely calculated. In this case, two double-humped solitons with nearly equal velocities and amplitudes can be observed. In the same year, we also studied the generalized NLS equation by IST \cite{xe2019}. In this paper, the high-order rogue wave of generalized NLS equation with nonzero boundary was given based on the robust IST method. This method is more convenient than before because we don't have to take a limit. A study of high-order solitons in three nonlocal NLS equations including the PT-symmetric, reverse-time, and reverse-space-time was presented in 2018 \cite{yb}. General high-order solitons in three different equations were derived from the same Riemann-Hilbert solutions of the AKNS hierarchy, except for the difference in the corresponding symmetry relations on the "perturbed" scattering data. Dynamics of general high-order solitons in these equations were further analyzed. It was shown that the high-order fundamental-soliton is moving on different trajectories in nearly equal velocities, and they can be nonsingular or repeatedly collapsing, depending on the choices of the parameters. It was also shown that the high-order multi-solitons could have more complicated wave structures and behaviors which are different from higher-order fundamental solitons.

There is also a wide range of literature concerning the behavior of solitons and their interactions in various integrable systems such as soliton scattering, breather solutions, and soliton bound states have been published recently \cite{xt32020, xt22019, zys22020, wxy2020}. As is known to all, some integrable nonlinear PDEs in mathematical physics have rich mathematical structures and extensive physics applications \cite{yh2020, pt2019}. In particular, it is always possible to find explicit solutions to these equations, such as they often have multi-soliton solutions. Among these integral PDEs, the NLS equation
\begin{equation}\label{nls}
iu_{t}+u_{xx}+2|u|^2u=0,
\end{equation}
has been considered as the most important mathematical model. Eq\eqref{nls} equations can be used to describe wave evolution in scientific fields such as water waves \cite{be1967, za1968}, plasma physics \cite{ha1972}, condensed matter physics, fluids, arterial mechanics and fiber optics \cite{ha19731, ha19732}. However, several phenomena observed in the experiment cannot be explained by NLS equation, as the short soliton pulses get shorter, some additional effects become important. The NLS-type equations with high-order terms have important effects
in fiber optics, Heisenberg spin chain and ocean waves. In order to describe the dynamics of a one-dimensional continuum anisotropic Heisenberg ferromagnetic spin chain with the octuple-dipole interaction or the alpha helical protein with higher-order excitation and interaction under the continuum approximation, an ENLS equation with higher-order odd (third-order) and even (fourth-order)terms has been studied \cite{enlseq2014}. The ENLS equation is as follows:
\begin{equation}\begin{array}{c}\label{enls}
iu_{t}+\frac{1}{2} u_{x x}+|u|^{2} u-i \alpha\left(u_{x x x}+6 u_{x}|u|^{2}\right) \\
+\gamma\left(u_{x x x x}+6 u_{x}^{2} u^{*}+4 u\left|u_{x}\right|^{2}+8 u_{x x}|u|^{2}+2 u_{x x}^{*} u^{2}+6 u|u|^{4}\right)=0.
\end{array}\end{equation}

Here, $t$ is the propagation variable and $x$ is the retarded time in the moving frame, with the function $u(x, t)$ being the envelope of the wave field. The notation is standard in the theory of nonlinear waves. Sometimes $x$ and $t$ are interchanged in optics and water wave
theory. All coefficients in this equation are fixed except for the $\alpha$ and $\gamma$.  The coefficients $\alpha$ and $\gamma$ are two real parameters which control independently the values of third-order dispersion $u_{x x x}$ and that of fourth-order dispersion $u_{x x x x}$. When coefficients $\alpha$ and $\gamma$ are equal to zero, the remaining part is the standard normalized NLS equation. If $\alpha\neq 0, \gamma=0,$ the equation is integrable and is known as the Hirota equation. Furthermore, when $\alpha=0, \gamma\neq0$ the equation is also integrable and known as the Lakshmanan-Porsezian-Daniel (LPD) equation. Ankiewicz and Akhmediev have indicated the integrability and derived the soliton solutions of \eqref{enls} by DT, which motivates us to search and analyze more exact solutions. To the best of our knowledge, the high-order solitons of the ENLS equation have never been reported.

The main subject of the present paper is to research the high-order solitons of the ENLS equation in the framework of the RHP. Through a standard dressing procedure, we can find the soliton matrix for the nonregular RHP with simple zeros. Then combined with gDT, soliton matrix for elementary high-order zeros in the RHP for the ENLS equation are constructed.  Moreover, the influence of free parameter ($\alpha$, $\gamma$) in soliton solutions of ENLS equation on soliton propagation, collision dynamics along with the asymptotic behavior for the two-solitons and long-time asymptotic estimations for the high-order one-soliton are concretely analyzed.  The propagation direction, velocity, width and other physical quantities of solitons can be modulated by adjusting the free parameters of ENLS equation. Our work may be helpful to observe the light pulse waves in optical fibers and guide optical experiments.

This paper is organized as follows. In Section 2, the inverse scattering theory is established for the $2\times2$ spectral problems, and the corresponding matrix RHP is formulated. In Section 3, the N-soliton formula for ENLS equation is derived by considering the simple zeros in the RHP. In Section 4, the high-order soliton matrix and the generalized DT is constructed and the explicit high-order N-soliton formula is obtained, which corresponds to the elementary high-order zeros in the RHP. The final section is devoted to conclusion and discussion.

\section{Inverse scattering theory for ENLS equation}

In this section, we consider the scattering and inverse scattering problem for ENLS equation.

{\bf \subsection{Scattering theory of the spectral problem}}

Considering the spectral problem of the ENLS equation \eqref{enls}:
\begin{equation}\label{b}
Y_{x}=LY,
\end{equation}
\begin{equation}\label{c}
Y_{t}=BY,
\end{equation}
with $2\times2$ matrices $L$ and $B$ in the forms of:
\begin{equation}\notag
\begin{split}
&L=-i\zeta\Lambda+Q,\\
&B=(-i\zeta^{2}+4i\alpha\zeta^{3}+8i\gamma\zeta^4)\Lambda+V_{1},\\
&V_{1}=(\zeta-4\alpha\zeta^2)Q+(\frac{1}{2}-2\alpha\zeta)V-\alpha K+\gamma V_{p}.
\end{split}
\end{equation}
Where
\begin{equation}\notag
\Lambda=\left(\begin{array}{cc}
1 & 0 \\
0 & -1
\end{array}\right),
\end{equation}
\begin{equation}\label{2.8}
Q=\left(\begin{array}{cc}
0 & u \\
-u^{*} & 0
\end{array}\right),
\end{equation}
\begin{equation}\notag
V=\left(\begin{array}{cc}
\mathrm{i}|u|^{2} & \mathrm{i} u_{x} \\
\mathrm{i} u_{x}^{*} & -\mathrm{i}|u|^{2}
\end{array}\right),
\end{equation}
\begin{equation}\notag
\begin{split}
&K=\left(\begin{array}{cc}u u_{x}^{*}-u^{*} u_{x} & -\left(2|u|^{2} u+u_{x x}\right) \\ 2|u|^{2} u^{*}+u_{x x}^{*} & -\left(u u_{x}^{*}-u^{*}u_{x}\right)\end{array}\right),\\
&V_{p}=\left(\begin{array}{cc}\mathrm{i} A_{p}(x, t) & B_{p}(x, t) \\ -B_{p}^{*}(x, t) & -\mathrm{i} A_{p}(x, t)\end{array}\right),
\end{split}
\end{equation}
with
\begin{equation}\notag
\begin{split}
\begin{array}{l}
A_{p}(x, t)=3|u|^{4}-\left|u_{x}\right|^{2}+u u_{x x}^{*}+u^{*} u_{x x}-2 \mathrm{i} \zeta\left(u^{*} u_{x}-u u_{x}^{*}\right)-4 \zeta^{2}|u|^{2}, \\
B_{p}(x, t)=6 \mathrm{i}|u|^{2} u_{x}+\mathrm{i} u_{x x x}+2 \zeta u_{x x}+4 \zeta|u|^{2} u-4 \mathrm{i} \zeta^{2} u_{x}-8 \zeta^{3} u.
\end{array}
\end{split}
\end{equation}
Here, $\zeta$ is a spectral parameter, $Y(x, t, \zeta)$ is a vector  function, and the superscript ``*" represents complex conjugation. The spatial linear operator \eqref{b}  and the temporal linear operator \eqref{c}  are the Lax pair of the ENLS equation \eqref{enls}.
 Supposing $u(x)=u(x, 0) \rightarrow 0$ sufficiently fast as $x \rightarrow \pm\infty$. For
a prescribed initial condition $u(x, 0)$, we seek the solution $u(x, t)$ at any later time t. That is, we solve an initial value problem for the ENLS equation.

Notation
\begin{equation}
E_{1}=e^{-i\zeta\Lambda x-(i\zeta^{2}-4i\alpha\zeta^3-8i\gamma\zeta^4)\Lambda t},
\end{equation}
\begin{equation}\label{2.5}
J=YE_{1}^{-1},
\end{equation}
so that the new matrix function $J$ is $(x, t)$-independent at infinity. Inserting \eqref{2.5} into \eqref{b}-\eqref{c}, we find that the Lax pair \eqref{b}-\eqref{c} becomes
\begin{equation}\label{2.6}
J_{x}=-i\zeta[\Lambda,J]+QJ,
\end{equation}
\begin{equation}\label{2.7}
J_{t}=-(i\zeta^{2}-4i\alpha\zeta^3-8i\gamma\zeta^4)[\Lambda,J]+V_{1}J,
\end{equation}
where $[\Lambda,J]=\Lambda J-J\Lambda$ is the commutator. Notice that both matrices $Q$ and $V_{1}$ are anti-Hermitian, i.e.
\begin{equation}\label{2.9}
Q^{\dag}=-Q, \ \ V_{1}^{\dag}=-V_{1},
\end{equation}
where the superscript $``\dag " $ represents the Hermitian of a matrix. In addition, their traces are both equal to zero, i.e. $trQ = trV_{1} = 0$.

Now we let time t be fixed and is a dummy variable, and thus it will be suppressed in our notation. In the scattering problem, we first introduce matrix Jost solutions $J_{\pm}(x,\zeta)$ of \eqref{2.6} with the following asymptotic at large distances:
\begin{equation}\label{2.10}
J_{\pm}(x,\zeta)\rightarrow I, \ \ \ x\rightarrow\pm\infty.
\end{equation}
Here, $I$ is the $2\times2$ unit matrix. Now we will delineate Jost solutions $J_{\pm}(x,\zeta)$  analytical properties first.

Introducing the notation $E=e^{-i\zeta \Lambda x}$, $\Phi\equiv J_{-}E$ and $\Psi \equiv J_{+}E$. $(\Phi, \Psi)$  satisfy the scattering equation \eqref{b},i.e.
\begin{equation}\label{2.19}
Y_{x}+i\zeta\Lambda Y=QY.
\end{equation}
When we treat the $QY$ term  as an inhomogeneous term,  $E$ is the solution to the homogeneous equation $Y_{x}+i\zeta\Lambda Y=0$, then using the method of variation of parameters as well as the boundary conditions \eqref{2.10}, we can turn \eqref{2.19} into
the following Volterra integral equations:
\begin{equation}\label{2.20}
J_{\pm}(x,\zeta)=I+\int_{\pm\infty}^{x}e^{-i\zeta\Lambda(x-y)}Q(y)J_{\pm}(y,\zeta)e^{-i\zeta\Lambda(y-x)}dy,
\end{equation}

Thus $J_{\pm}(x,\zeta)$ allow analytical continuations off the real axis $\zeta \in R$ as long as the integrals on the right sides of the above Volterra equations converge. Due to the structure \eqref{2.8} of the potential $Q$, we can easily get the following proposition.

\begin{proposition}\label{p1}
 The first column of $J_{-}$ and the second column of $J_{+}$
can be analytically continued to the upper half plane $\zeta \in \mathbb{C}_{+}$ , while the second column of
$J_{-}$ and the first column of $J_{+}$ can be analytically continued to the lower half plane $\mathbb{C}_{-}$.
\end{proposition}
\begin{proof}
The integral equation \eqref{2.20} for the first column of $J_{-}$ , say $(\varphi_{1},\varphi_{2})^{T}$, is
\begin{equation}\label{2.22}
\varphi_{1}=1+\int_{-\infty}^{x}u(y)\varphi_{2}(y,\zeta)dy,
\end{equation}
\begin{equation}\label{2.23}
\varphi_{2}=-\int_{-\infty}^{x}u^{*}(y)\varphi_{1}(y,\zeta)e^ {2i\zeta(x-y)}dy.
\end{equation}

When $\zeta \in \mathbb{C}_{+}$, since $e^{2i\zeta (x-y)}$ in \eqref{2.23} is bounded, and $u(x)$ decays to zero sufficiently fast at large distances, both integrals in the above two equations converge. Thus the Jost solution $(\varphi_{1}, \varphi_{2})^{T}$ can be analytically extended to $\mathbb{C}_{+}$. The analytic properties of the other Jost solutions $J_{+}$ can be obtained similarly.
\end{proof}

From Abel's identity, we find that $|J(x, \zeta)|$ is a constant for all $x$. Then using the boundary conditions \eqref{2.10}, we see that
\begin{equation}\label{2.13}
|J_{\pm}(x,\zeta)|=1,
\end{equation}
for all $(x, \zeta)$. Since $\Phi(x, \zeta)$ and $\Psi(x, \zeta)$ are both solutions of the linear equation \eqref{b}, they are linearly related by a scattering matrix $S(\zeta)=(s_{ij})_{2\times2}$:
\begin{equation}\label{2.16}
\Phi(x, \zeta)=\Psi(x, \zeta)S(\zeta),\quad \zeta\in \mathbb{R}.
\end{equation}
$\mathbb{R}$ is the set of real numbers.

Because we need to use scattering matrix $S(\zeta)$ to reconstruct the potential $u(x,t)$, now we need to delineate the analytical properties of $S(\zeta)$. If we express $(\Phi,\Psi)$ as a collection of columns
\begin{equation}\label{2.24}
\Phi=(\phi_{1}^{+}, \phi_{2}^{-}),\quad  \Psi=(\psi_{1}^{-}, \psi_{2}^{+}).
\end{equation}

Where the superscripts $``\pm"$ indicate the half plane of analyticity for the underlying quantities.
 Since
\begin{equation}\label{2.40}
S=\Psi^{-1}\Phi=\left(\begin{matrix}
\hat{\psi_{1}}^{+}\\
\hat{\psi_{2}}^{-}\end{matrix}
\right)(\phi_{1}^{+},\phi_{2}^{-}),
\end{equation}
\begin{equation}\label{2.41}
S^{-1}=\Phi^{-1}\Psi=\left(\begin{matrix}
\hat{\phi_{1}}^{-}\\
\hat{\psi_{2}}^{+}\end{matrix}
\right)(\psi_{1}^{-},\psi_{2}^{+}).
\end{equation}
We see immediately that scattering matrices $S$ and $S^{-1}$ have the following analyticity structures:
\begin{equation}\notag
S=\left(\begin{matrix}
s_{11}^{+},s_{12}\\s_{21},s_{22}^{-}\end{matrix}
\right), \quad
S^{-1}=\left(\begin{matrix}
\hat{s}_{11}^{-},\hat{s}_{12}\\ \hat{s}_{21},\hat{s}_{22}^{+}\end{matrix}
\right).
\end{equation}
Elements without superscripts indicate that such elements do not
allow analytical extensions to $\mathbb{C}_{\pm}$ in general. Because $S$ is $2\times2$
matrix with unit determinant, then we can get
\begin{equation}
\hat{s}_{11} = s_{22},\quad  \hat{s}_{22} =s_{11},\quad \hat{s}_{12}=-s_{12},\quad  \hat{s}_{21} =-s_{21}.
\end{equation}
Hence analytic properties of $S^{-1}$ can be directly read off from analytic properties of $S$.

In order to construct the RHP, we define the Jost solutions
\begin{equation}\label{2.25}
P^{+}=(\phi_{1},\psi_{2})e^{i\zeta\Lambda x}=J_{-}H_{1}+J_{+}H_{2}
\end{equation}
are analytic in $\zeta \in \mathbb{C}_{+}$, here
\begin{equation}\label{2.27}
H_{1}\equiv diag(1,0),\quad H_{2}\equiv diag(0,1).
\end{equation}
In addition, from the Volterra integral equations \eqref{2.20}, we see that the large $\zeta$ asymptotics of these analytical functions are
\begin{equation}\label{2.28}
P^{+}(x,\zeta)\rightarrow I, \zeta \in \mathbb{C}_{+}\rightarrow\infty,
\end{equation}
If we express $\Phi^{-1}$ and $\Psi^{-1}$ as a collection of rows
\begin{equation}\label{2.33}
\Phi^{-1}=\left(\begin{matrix}
\hat{\phi_{1}}\\
\hat{\phi_{2}}\end{matrix}
\right),\ \ \ \ \ \
\Psi^{-1}=\left(\begin{matrix}
\hat{\psi_{1}}\\
\hat{\psi_{2}}\end{matrix}
\right).
\end{equation}
Then by techniques similar to those used above, we can show that the adjoint Jost solutions
\begin{equation}\label{2.34}
P^{-}= e^{-i\zeta\Lambda x}\left(\begin{matrix}
\hat{\phi_{1}}\\
\hat{\psi_{2}}\end{matrix}
\right)=H_{1}J_{-}^{-1}+H_{2}J_{+}^{-1}
\end{equation}
are analytic in $\zeta \in \mathbb{C}_{-}$. In addition,
\begin{equation}\label{2.36}
P^{-}(x,\zeta)\rightarrow I, \zeta \in \mathbb{C}_{-}\rightarrow \infty.
\end{equation}
The anti-Hermitian property \eqref{2.9} of the potential matrix $Q$ gives rise to involution properties in the scattering matrix as well as in the Jost solutions. Indeed, by taking the Hermitian of the scattering equation \eqref{2.6} and utilizing the anti-Hermitian property of the potential matrix $Q^{\dag}=-Q$, we get
\begin{equation}
J^{\dag}_{\pm}(\zeta^{*}) = J^{-1}_{\pm}(\zeta).
\end{equation}
From this involution property as well as the definitions \eqref{2.25} and \eqref{2.34} for $P^{\pm}$ we see that the analytic solutions $P^{\pm}$ satisfy the involution property as well:
\begin{equation}\label{2.57}
(P^{+})^{\dag}(\zeta^{*}) = P^{-}(\zeta).
\end{equation}
In addition, in view of the scattering relation \eqref{2.16} between $J_{+}$ and $J_{-}$, we see that $S$ also satisfies the involution property:
\begin{equation}\label{2.58}
S^{\dag}(\zeta^{*}) = S^{-1}(\zeta).
\end{equation}

{\bf \subsection{Matrix Riemann-Hilbert problem}}

Hence we have constructed two matrices functions $P^{\pm}(x, \zeta)$ which are analytic for $\zeta$ in $\mathbb{C}_{\pm}$, respectively. On the real line, using \eqref{2.16}, \eqref{2.25} and \eqref{2.34}, we easily get
\begin{equation}\label{2.44}
P^{-}(x,\zeta)P^{+}(x,\zeta)= G(x,\zeta),\ \ \ \zeta\in \mathbb{R},
\end{equation}
where
\begin{equation}
G = E(H_{1} +H_{2}S)(H_{1}+S^{-1}H_{2})E^{-1}=E\left(\begin{matrix}
1 \ \ \,\hat{s}_{12}\\s_{21}\ \ 1\end{matrix}
\right)E^{-1}.
\end{equation}

Equation \eqref{2.44} forms a matrix RHP. The normalization condition for this RHP can be obtained from \eqref{2.28} and \eqref{2.36} as
\begin{equation}\label{2.46}
P^{\pm}(x,\zeta)\rightarrow I, \ \ \ \zeta \rightarrow \infty,
\end{equation}
which is the canonical normalization condition.

Recalling the definitions \eqref{2.25} and \eqref{2.34} of $P^{\pm}$ as well as the scattering relation \eqref{2.16}, we see that
\begin{equation}\label{2.52}
|P^{+}|= \widehat{s}_{22} = s_{11},\quad |P^{-}|= s_{22} = \widehat{s}_{11} .
\end{equation}

We consider the solution of the regular RHP first, i.e. : $|P^{\pm}| \neq 0$, i.e. :$\widehat{s}_{22} = s_{11}\neq0$ and $s_{22} = \widehat{s}_{11}\neq0$ in their respective planes of analyticity. Under the canonical normalization condition \eqref{2.46}, the solution
to this regular RHP is unique\cite{yjks}. This unique solution to the regular matrix RHP \eqref{2.44} defies explicit expressions. But its formal solution can be given in terms of a Fredholm integral equation.
To use this Plemelj-Sokhotski formula on the regular RHP \eqref{2.44}, we first rewrite \eqref{2.44} as
\[
\left(P^{+}\right)^{-1}(\zeta)-P^{-}(\zeta)=\widehat{G}(\zeta)\left(P^{+}\right)^{-1}(\zeta), \quad \zeta \in \mathbb{R},
\]
where
\[
\widehat{G}=I-G=-E\left(\begin{array}{cc}
0 & \hat{s}_{12} \\
s_{21} & 0
\end{array}\right) E^{-1}.
\]

$\left(P^{+}\right)^{-1}(\zeta)$ is analytic in $\mathbb{C}_{+},$ and $P^{-}(\zeta)$ is analytic in $\mathbb{C}_{-}$.Applying the Plemelj-Sokhotski formula and utilizing the canonical boundary conditions \eqref{2.46}, the solution to the
regular RHP \eqref{2.44} is provided by the following integral equation:
\[
\left(P^{+}\right)^{-1}(\zeta)=I+\frac{1}{2 \pi i} \int_{-\infty}^{\infty} \frac{\widehat{G}(\xi)\left(P^{+}\right)^{-1}(\xi)}{\xi-\zeta} d \xi, \quad \zeta \in \mathbb{C}_{+}.
\]

In the more general case, the RHP \eqref{2.44} is not regular, i.e. $|P^{+}(\zeta)|$ and $|P^{-}(\zeta)|$ can be zero at certain discrete locations $\zeta_{k} \in \mathbb{C}_{+}$ and $\bar{\zeta}_{k} \in \mathbb{C}_{-}, 1 \leq k \leq N$, where $N$ is the number of these zeros. In view of \eqref{2.52}, we see that $\left(\zeta_{k}, \bar{\zeta}_{k}\right)$ are zeros of the scattering coefficients $\hat{s}_{22}(\zeta)$ and $s_{22}(\zeta)$. Due to the involution property \eqref{2.58}, we have the involution relation
\begin{equation}\label{2.66}
\bar{\zeta}_{k}=\zeta_{k}^{*}.
\end{equation}
For simplicity, we assume that all zeros $\left\{\left(\zeta_{k}, \bar{\zeta}_{k}\right), k=1, \ldots, N\right\}$ are simple zeros of $\left(\hat{s}_{22}, s_{22}\right)$ which is the generic case. In this case, both ker($P^{+}\left(\zeta_{k}\right)$) and ker($P^{-}\left(\bar{\zeta}_{k}\right)$) are spanned by one-dimensional column vector $\left|v_{k}\right\rangle$ and row vector $\left\langle v_{k}\right|$, respectively.
\begin{equation}\label{2.67}
P^{+}\left(\zeta_{k}\right) \left|v_{k}\right\rangle=0, \quad \left\langle v_{k}\right| P^{-}\left(\bar{\zeta}_{k}\right)=0, \quad 1 \leq k \leq N.
\end{equation}
Taking the Hermitian of the first equation in \eqref{2.67} and utilizing the involution properties \eqref{2.57} and \eqref{2.66},
we see that eigenvectors $\left(\left|v_{k}\right\rangle, \left\langle v_{k}\right|\right)$
satisfy the involution property $\left\langle v_{k}\right|=\left|v_{k}\right\rangle^{\dagger}$, vectors $\left|v_{k}\right\rangle$ and $\left\langle v_{k}\right|$ are $x$ dependent, our starting point is \eqref{2.67} for $\left|v_{k}\right\rangle$ and $\left\langle v_{k}\right|$. Taking the $x$ derivative to the $\left|v_{k}\right\rangle$ equation and recalling that $P^{+}$ satisfies the scattering equation \eqref{2.6}, we get
\begin{equation}\label{2.99}
\left|v_{k}(x)\right\rangle=e^{-i\zeta_{k} \Lambda x}\left|v_{k0}\right\rangle,
\end{equation}
where $\left|v_{k0}\right\rangle=\left|v_{k}(x)\right\rangle|_{x=0}.$

Following similar calculations for $\bar{v}_{k},$ we readily get
\[
\left\langle v_{k}(x)\right|=\left\langle \bar{v}_{k0}\right| e^{i \zeta_{k} \Lambda x}.
\]
These two equations give the simple $x$ dependence of vectors $\left|v_{k}(x)\right\rangle$ and $\left\langle \bar{v}_{k}(x)\right|$.
The zeros $\left\{\left(\zeta_{k}, \bar{\zeta}_{k}\right)\right\}$ of $|P^{\pm}(\zeta)|$ as well as vectors ${\left|v_{k}\right\rangle, \left\langle v_{k}\right|}$ in the kernels of $P^{+}\left(\zeta_{k}\right)$ and $P^{-}\left(\bar{\zeta}_{k}\right)$ constitute the discrete scattering data which is also needed to solve the general RHP \eqref{2.44}.

Now we construct a matrix function which could remove all the zeros of this RHP. For this purpose,
we will introduce the rational matrix function:
\[
\Gamma_{j}=I+\frac{\bar{\zeta}_{j}-\zeta_{j}}{\zeta-\bar{\zeta}_{j}} \frac{\left|v_{j}\right\rangle\left\langle v_{j}\right|}{\left\langle v_{j} | v_{j}\right\rangle},
\]
and its inverse matrix
\[
\Gamma_{j}^{-1}=I+\frac{\zeta_{j}-\bar{\zeta}_{j}}{\zeta-\zeta_{j}} \frac{\left|v_{j}\right\rangle\left\langle v_{j}\right|}{\left\langle v_{j} | v_{j}\right\rangle},
\]
where
\[
\begin{array}{l}
\left|v_{i}\right\rangle \in \operatorname{Ker}\left(P^{+} \Gamma_{1}^{-1} \cdots \Gamma_{i-1}^{-1}\left(\zeta_{i}\right)\right),
\quad \left\langle v_{j}|=| v_{j}\right\rangle^{\dagger}.
\end{array}
\]
Therefore, if one is introducing the matrix function:
\[
\Gamma=\Gamma_{N}\Gamma_{N-1} \cdots \Gamma_{1},
\]
\[
\begin{aligned}\label{2.72}
\Gamma(\zeta) &=I+\sum_{j, k=1}^{N} \frac{\left|v_{j}\right\rangle\left(M^{-1}\right)_{j k}\left\langle v_{k}\right|}{\zeta-\bar{\zeta}_{k}}
\end{aligned},
\]
\[
\begin{aligned}\label{2.73}
\Gamma^{-1}(\zeta) &=I-\sum_{j, k=1}^{N} \frac{\left|v_{j}\right\rangle\left(M^{-1}\right)_{j k}\left\langle v_{k}\right|}{\zeta-\zeta_{j}}
\end{aligned}.
\]

$M$ is an $N \times N$ matrix with its $(j, k)$ th element given by
\begin{equation}\label{2.741}
M_{j k}=\frac{\left\langle v_{j}|v_{k}\right\rangle}{\bar{\zeta}_{j}-\zeta_{k}}, \quad 1 \leq j, k \leq N, \\
\end{equation}
then $\Gamma(x, \zeta)$ cancels all the zeros of $P_{\pm},$ and the analytic solutions can be represented as
\[
\begin{array}{l}
P^{+}(\zeta)=\widehat{P}^{+}(\zeta) \Gamma(\zeta), \\
P^{-}(\zeta)=\Gamma^{-1}(\zeta) \widehat{P}^{-}(\zeta).
\end{array}
\]
Here, $\widehat{P}^{\pm}(\zeta)$ are meromorphic $2 \times 2$ matrix functions in $\mathbb{C}_{+}$ and $\mathbb{C}_{-}$, respectively, with finite number of poles and specified residues. Therefore, all the zeros of RHP have been eliminated and we can formulate a regular RHP
\[
\widehat{P}^{-}(\zeta) \widehat{P}^{+}(\zeta)=\Gamma(\zeta) G(\zeta) \Gamma^{-1}(\zeta), \quad \zeta \in \mathbb{R},
\]
with boundary condition:  $\widehat{P}^{\pm}(\zeta)=P^{\pm}(\zeta)\Gamma^{-1} \rightarrow I$ as $\zeta \rightarrow \infty$. As a result, when $\zeta \rightarrow \infty$ we have $ P^{+}(\zeta)=\Gamma.$

{\bf \subsection{Solution of the Riemann-Hilbert Problem}}

In this subsection, we discuss how to solve the matrix RHP \eqref{2.44} in the complex $\zeta$ plane. This inverse problem can be solved by expanding $P^{\pm}$ at large $\zeta$ as
\begin{equation}\label{2.47}
P^{\pm}(x,\zeta) = I +\zeta^{-1}P_{1}^{\pm}(x)+O(\zeta^{-2} ), \quad \zeta\rightarrow\infty,
\end{equation}
and insert \eqref{2.47} into \eqref{2.6}, then by comparing terms of the same power in $\zeta^{-1}$, $\zeta^{0}$. We found that
\begin{equation}\label{2.50}
diag(P^{+}_{1})_{x} = diag(QP^{+}_{1}).
\end{equation}
\begin{equation}\label{2.48}
Q = i[\wedge,P_{1}^{+}]=-i[\wedge,P_{1}^{-}].
\end{equation}
Hence the solution $u$ can be reconstructed by
\begin{equation}\label{2.49}
u = 2i( P_{1}^{+})_{12}= -2i(P_{1}^{-})_{12}.
\end{equation}
This completes the inverse scattering process. How to solve the matrix RHP \eqref{2.44} will be discussed in the next subsection.

{\bf \subsection{Time Evolution of Scattering Data}}

In this subsection, we determine the time evolution of the scattering data. First we determine the time evolution of the scattering matrices $S$ and $S^{-1}$. Our starting point is the definition \eqref{2.16} for the scattering matrix, which can be rewritten as
\[
J_{-} E=J_{+} E S, \quad \zeta \in \mathbb{R}.
\]

Since $J_{\pm}$ satisfies the temporal equation \eqref{2.7} of the Lax pair, then multiplying \eqref{2.7} by the time-independent diagonal matrix $E=e^{-i \zeta \Lambda x},$ we see that $J_{-} E,$ i.e. $J_{+} E S,$ satisfies the same temporal equation \eqref{2.7} as well. Thus, by inserting $J_{+} E S$ into \eqref{2.7}, taking the limit $x \rightarrow+\infty,$ and recalling the boundary condition \eqref{2.10} for $J_{+}$ as well as the fact that $V \rightarrow 0$
as $x \rightarrow \pm \infty,$ we get
\[
S_{t}=-(i\zeta^{2}-4i\alpha\zeta^{3}-8i\gamma\zeta^{4})[\Lambda, S].
\]
Similarly, by inserting $J_{-} E S^{-1}$ into \eqref{2.7}, taking the limit $x \rightarrow-\infty,$ and recalling the asymptotics \eqref{2.10} for $J_{-},$ we get
\[\label{2.103}
\left(S^{-1}\right)_{t}=-(i\zeta^{2}-4i\alpha\zeta^{3}-8i\gamma\zeta^{4})\left[\Lambda, S^{-1}\right].
\]
From these two equations, we get
\begin{equation}\label{2.104}
\frac{\partial \hat{s}_{22}}{\partial t}=\frac{\partial s_{22}}{\partial t}=0,
\end{equation}
and
\begin{equation}\label{2.105}
\frac{\partial \hat{s}_{12}}{\partial t}=-(i\zeta^{2}-4i\alpha\zeta^{3}-8i\gamma\zeta^{4}) \hat{s}_{12}, \quad \frac{\partial s_{21}}{\partial t}=(i\zeta^{2}-4i\alpha\zeta^{3}-8i\gamma\zeta^{4}) s_{21}.
\end{equation}
The two equations in \eqref{2.104} show that $\hat{s}_{22}$ and $s_{22}$ are time independent. Recall that $\zeta_{k}$ and $\bar{\zeta}_{k}$ are zeros of $|P^{\pm}(\zeta)|$, i.e. they are zeros of $\hat{s}_{22}(\zeta)$ and $s_{22}(\zeta)$ in view of \eqref{2.52}. Thus $\zeta_{k}$ and $\bar{\zeta}_{k}$ are also time independent. The two equations in \eqref{2.105} give the time evolution for the scattering data $\hat{s}_{12}$ and $s_{21},$ which is
\[
\hat{s}_{12}(t;\zeta)=\hat{s}_{12}(0 ; \zeta) e^{-(i\zeta^{2}-4i\alpha\zeta^{3}-8i\gamma\zeta^{4})t}, \quad s_{21}(t ; \zeta)=s_{21}(0 ; \zeta) e^{(i\zeta^{2}-4i\alpha\zeta^{3}-8i\gamma\zeta^{4}) t}.
\]
Next we determine the time dependence of the scattering data $ \left|v_{k}\right\rangle$ and $\left\langle v_{j}\right|$. We start with \eqref{2.67} for $\left|v_{k}\right\rangle$ and $\left\langle v_{j}\right|$. Taking the time derivative to the $\left|v_{k}\right\rangle$ equation and recalling that $P^{+}$ satisfies the temporal equation $\eqref{2.67}$, we get
\[
P^{+}\left(\zeta_{k} ; x, t\right)\left(\frac{\partial \left|v_{k}\right\rangle}{\partial t}+(i\zeta^{2}-4i\alpha\zeta^{3}-8i\gamma\zeta^{4})\Lambda \left|v_{k}\right\rangle\right)=0,
\]
thus
\[
\frac{\partial \left|v_{k}\right\rangle}{\partial t}+(i\zeta^{2}-4i\alpha\zeta^{3}-8i\gamma\zeta^{4})\left|v_{k}\right\rangle=0.
\]
Combining it with the spatial dependence \eqref{2.99}, we get the temporal and spatial dependence for the vector $\left|v_{k}\right\rangle$ as
\begin{equation}\label{2.109}
\left|v_{k}\right\rangle(x, t)=e^{-i \zeta \Lambda x-(i\zeta^{2}-4i\alpha\zeta^{3}-8i\gamma\zeta^{4})\Lambda t} \left|v_{k0}\right\rangle,
\end{equation}
where $\left|v_{k0}\right\rangle$ is a constant. Similar calculations for $\left\langle v_{k}\right|$ give
\[\left\langle v_{k}\right|(x, t)=\left\langle v_{k0}\right| e^{i \bar{\zeta_{k}}x+(i\bar{\zeta}^{2}-4i\alpha\bar{\zeta}^{3}-8i\gamma\bar{\zeta}^{4})t}.\]
 We see that the scattering data needed to solve this non-regular RHP is
\begin{equation}\label{2.95}
\left\{s_{21}(\xi), \hat{s}_{12}(\xi), \xi \in \mathbb{R} ; \quad \zeta_{k}, \bar{\zeta}_{k}, \left|v_{j}\right\rangle, \left\langle v_{k}\right|, 1 \leq k \leq N\right\}.
\end{equation}
This is called minimal scattering data. From this scattering data at any later time, we can solve the non-regular RHP \eqref{2.44} with zeros \eqref{2.67}, and thus reconstruct the solution $u(x, t)$ at any later time from the formula \eqref{2.49}. So far, the IST process for ENLS equation \eqref{enls} has been completed.

\section{N-Soliton Solutions}

It is well known that when scattering data $\hat{s}_{12}=s_{21}=0$, the soliton solutions  correspond to the reflectionless potential. Then jump matrix $G = I$, $\widehat{G} = 0$. Due to $ P^{+}(\zeta)=\Gamma, \zeta\rightarrow\infty$. Recall to \eqref{2.49}, we can get
\begin{equation}\label{2.120}
u(x,t)=2i(\sum_{j,k=1}^{N}\left|v_{j}\right\rangle(M^{-1})_{jk}\left\langle v_{k}\right|)_{12}.
\end{equation}
Here vectors $\left|v_{j}\right\rangle$ are given by \eqref{2.109}, $\left\langle v_{k}\right|=\left|v_{k}\right\rangle^{\dag}$, and matrix $M$ is given by \eqref{2.741}. Without loss of generality, we let $\left|v_{k0}\right\rangle= (c_{k}, 1)^{T}$. In addition, we introduce the notation
\begin{equation}\label{2.121}
\theta_{k}=-i\zeta_{k}x-(i\zeta_{k}^{2}-4i\alpha\zeta_{k}^{3}-8i\gamma\zeta_{k}^{4})t.
\end{equation}
Then the above solution $u$ can be written out explicitly as
\begin{equation}\label{2.122}
u(x,t)=2i\sum_{j,k=1}^{N}c_{j}e^{\theta_{j}-\theta_{k}^{*}}(M^{-1})_{jk},
\end{equation}
where the elements of the $N \times N$ matrix M are given by
\begin{equation}\label{2.123}
M_{jk}=\frac{1}{\zeta_{j}^{*}-\zeta_{k}}[e^{-(\theta_{k}+\theta_{j}^{*})}
+c_{j}^{*}c_{k}e^{\theta_{k}+\theta_{j}^{*}}].
\end{equation}
Notice that $M^{-1}$ can be expressed as the transpose of $M's$ cofactor matrix divided by $|M|$. Also recall that the determinant of a matrix can be expressed as the sum of its elements along a row or column multiplying their corresponding cofactor. Hence the solution \eqref{2.122} can be rewritten as
\begin{equation}\label{2.124}
u(x,t) = -2i\frac{|F|}{|M|},
\end{equation}

where $F$ is the following $(N +1)\times(N +1)$ matrix:
\begin{equation}
\left(\begin{matrix}
0&e^{-\theta_{1}^{*}}&...&e^{-\theta_{N}^{*}}\\
c_{1}e^{\theta_{1}}&M_{11}&...&M_{N1}\\
.&.&.&.\\
.&.&.&.\\
.&.&.&.\\
c_{N}e^{\theta_{N}}&M_{1N}&...&M_{NN}
\end{matrix}
\right).
\end{equation}

Let $N = 1$, $\zeta1=\xi+i\eta$, $c1=1$ the solution \eqref{2.124} is
 \begin{equation}
 \begin{split}
u&={\frac {2\,i{c_{1}}\, ({\zeta}_{1}^{*}-{\it \zeta_{1}}) {{\rm e}^{-\theta_{1}^{*}+\theta_{1}}}}{(|{c_{1}}|) ^{2}{{\rm e}^{\theta_{1}^{*}+\theta_{1}}}+{
{\rm e}^{-\theta_{1}^{*}-\theta_{1}}}}}\\
&=2\eta\,{\rm sech} (2\,\eta\, A)\,{{\rm e}^{2\,i((8\,{\eta}^{4}\gamma-
48\,{\eta}^{2}\gamma\,{\xi}^{2}+8\,{\xi}^{4}\gamma-12\,\alpha\,{\eta}^
{2}\xi+4\,{\xi}^{3}\alpha+{\eta}^{2}-{\xi}^{2})t-x\xi)}},
\end{split}
\end{equation}
where
\begin{equation}\notag
A=((32\,{\eta}^{2}\gamma\,\xi-32\,
\gamma\,{\xi}^{3}+4\,\alpha\,{\eta}^{2}-12\,\alpha\,{\xi}^{2}+2\,
\xi)t+x).
\end{equation}

This solution is a solitary wave. Its amplitude function $|u|$ has the shape of a hyperbolic secant with peak amplitude $2\eta,$ and its velocity is $-(32\,{\eta}^{2}\gamma\,\xi-32\,\gamma\,{\xi}^{3}+4\,\alpha\,{\eta}^{2}-12\,\alpha\,{\xi}^{2}+2\,\xi).$ The phase of this solution depends linearly on both space $x$ and time $t.$ The spatial gradient of the phase is proportional to the speed of the wave. This solution is called a single-soliton solution of the ENLS equation\eqref{enls}.

Solve the equation $|u|^{2}=b$, $0<b<4\eta^{2}$, we can get
\begin{equation}
\begin{split}
x_{1}=\frac{-128\gamma\eta^3t\xi+128\gamma\eta t\xi^3-16\alpha\eta^3t+48\alpha\eta t\xi^2-8\eta t\xi+ln(\frac{8\eta^2+4\sqrt{4\eta^4-b\eta^2}-b}{b})}{4\eta}, \\
x_{2}=\frac{-128\gamma\eta^3t\xi+128\gamma\eta t\xi^3-16\alpha\eta^3t+48\alpha\eta t\xi^2-8\eta t\xi+ln(\frac{8\eta^2-4\sqrt{4\eta^4-b\eta^2}-b}{b})}{4\eta}.
\end{split}
\end{equation}
Notice $d$ is the width of the wave and
\begin{equation}
d=\frac{ln\frac{8\eta^2-4\sqrt{-\eta^2(-4\eta^2+b)}-b}{b}
-ln\frac{8\eta^2+4\sqrt{-\eta^2(-4\eta^2+b)}-b}{b})}{4\eta}.
\end{equation}
This means that the wave width is only related to the imaginary part of the spectral parameter and is not affected by the coefficients $\alpha$ and $\gamma$. The dispersion term and the nonlinear term in the higher order term of the ENLS equation play a good balance, which makes the system energy conservation.

Further, we can derive the center trajectory of the single-soliton solution
\begin{equation}
x=-(32\gamma\xi(\eta^2-\xi^2)+4\alpha(\eta^2-(\sqrt{3}\xi)^{2})+2\xi)t.
\end{equation}
The angle between the center trajectory and the t-axis is $arctan(-32\eta^2\gamma\xi+32\gamma\xi^3-4\alpha\eta^2+12\alpha\xi^2-2\xi)$.
Generally speaking, the third-order and  fourth-order coefficients $(\alpha, \gamma)$ affect both the velocity of the soliton and the slope of the central trajectory. But when the spectral parameter has only an imaginary part $(\xi=0)$ or $|\eta|=|\xi|$, the fourth-order coefficient $\gamma$ no longer affects the above quantities. When $|\eta|=|\sqrt{3}\xi|$, the third-order coefficient $\alpha$ no longer affects the above quantities. Without loss of generality, it can be divided into the following cases:\\
  case 1: $\eta=\xi=1.$ The velocity of the soliton and the slope of the central trajectory are equal to $8\alpha-2$,
so we can take three special cases: $\alpha=1/8, \alpha=1/4, \alpha=1/2.$\\
  csae 2: $\xi=0, \eta=1.$ The velocity of the soliton and the slope of the central trajectory are equal to $-4\alpha\eta^2$,
take three special cases: $\alpha=-1,$  $\alpha=0$,  $\alpha=1$.\\
 case 3: $\eta=1, \xi=\frac{1}{\sqrt{3}}$. The velocity of the soliton and the slope of the central trajectory are equal to $-(\frac{64\sqrt{3}}{9}\gamma+\frac{2\sqrt{3}}{3})$, take three special cases: $\gamma=-3/64, \gamma=-3/32, \gamma=-3/16.$\\
  case 4: $\eta=1, \xi=\frac{1}{2}$. The velocity of the soliton and the slope of the central trajectory are equal to $-(12\gamma+\alpha+1)$, fix $\alpha=-1$, and take three special cases: $\gamma=-1, \gamma=0, \gamma=1.$\\
We can control the propagation direction and speed of solitons by adjusting the parameters. The images of the central trajectory under different parameters is shown in Figure 1.

\begin{figure}
\centering
\includegraphics[width=4cm,height=4cm]{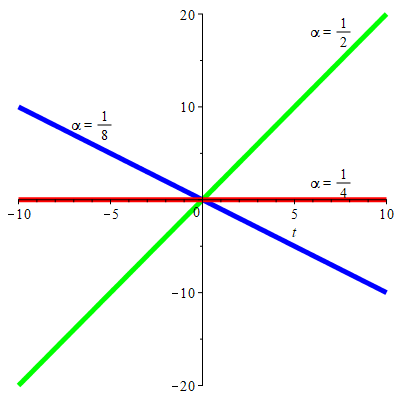}
$1$
\includegraphics[width=4cm,height=4cm]{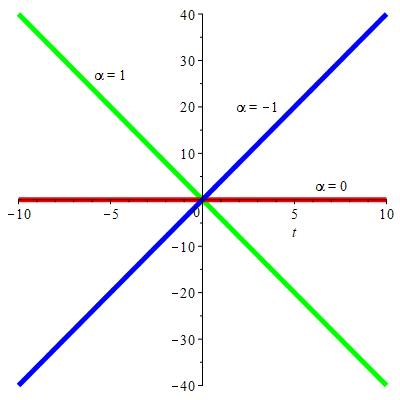}
$2$\\
\includegraphics[width=4cm,height=4cm]{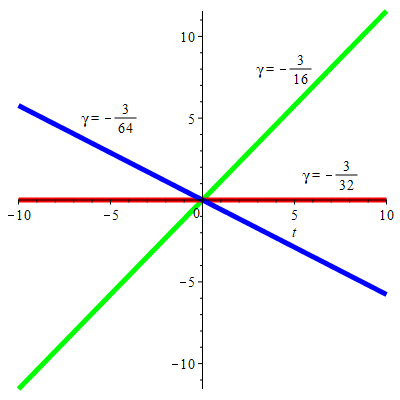}
$3$
\includegraphics[width=4cm,height=4cm]{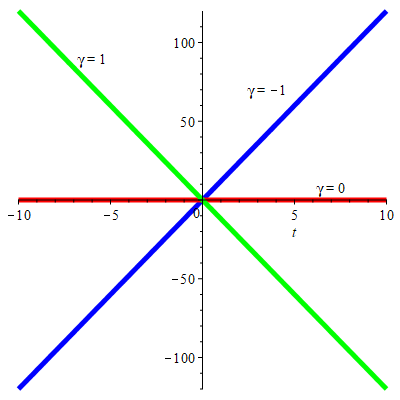}
$4$
\caption{central trajectory for the single-soliton solution of case 1, 2, 3 and 4.}
\label{t1}
\end{figure}

When $N=2$, the two-soliton solutions of ENLS equation can be written out explicitly as follows:
\begin{equation}\label{2.1241}\footnotesize%
\begin{split}
u=\frac{h_{1}\,{{\rm e}^{\Theta'_{1}+\Theta_{1}}}+{h_{2}}\,{{\rm e}^{\Theta'_{2}-\Theta_{2}}}+{h_{3}}\,{{\rm e}^{\Theta_{1}-\Theta'_{1}}}+{h_{4}}\,{{\rm e}^{\Theta'_{2}+\Theta_{2}}}
}{{d_{1}}\,{{\rm e}^{\Theta'_{1}-\Theta_{2}}}+{d_{2}}\,{{\rm e}^{\Theta'_{1}+\Theta_{2}}}+{d_{3}}\,{{\rm e}^{-\Theta'_{1}+\Theta_{2}}}+{d_{4}}\,{{\rm e}^{\Theta'_{1}-\Theta_{1}}}+{ d_{5}}\,{{\rm e}^{
-\Theta'_{1}-\Theta_{2}}}+{d_{6}}\,{
{\rm e}^{-\Theta'_{2}+\Theta_{1}}}},
\end{split}
\end{equation}
where
\begin{equation}\notag
\begin{split}
&d_{1}=(\zeta^{*}_{1}-\zeta^{*}_{2})({\zeta_{2}}-{\zeta_{1}}),\\
&d_{2}=|{c_{2}}|^{2}(\zeta^{*}_{1}-\zeta_{2})(\zeta^{*}_{2}-\zeta_{1}),\\
&d_{3}=|{{c_{1}}}^{2}{{c_{2}}}^{2}|(\zeta^{*}_{2}-\zeta^{*}_{1})({\zeta_{1}}-{\zeta_{2}}),\\ &d_{4}={c_{1}}{c_{2}}^{*}(\zeta^{*}_{1}-\zeta_{1})({\zeta_{2}}-{\zeta_{2}}^{*}),\\
&d_{5}=|{c_{1}}|^{2}(\zeta^{*}_{2}-\zeta_{1})(\zeta_{1}^{*}-{\zeta_{2}}),\\
&d_{6}=c_{1}^{*}{c_{2}}(\zeta^{*}_{2}-\zeta_{2})({\zeta_{1}}-\zeta_{1}^{*}),\\
&h_{1}=2i{c_{2}}\, ({\zeta_{2}}-\zeta^{*}_{1})(\zeta^{*}_{2}-\zeta^{*}_{1})(\zeta^{*}_{2}-{\zeta_{2}}),\\
&h_{2}=-2i{c_{1}}\,({\zeta_{1}}-\zeta^{*}_{1})(\zeta^{*}_{2}-\zeta^{*}_{1})(\zeta^{*}_{2}-{\zeta_{1}}),\\
&h_{3}=-2ic_{2}|{c_{1}}|^{2}(\zeta^{*}_{2}-{\zeta_{2}})(\zeta_{1}-\zeta_{2})({\zeta_{2}^{*}}-{\zeta_{1}}),\\
&h_{4}=2i{c_{1}}\,|{c_{2}}|^{2}({\zeta_{1}}-{\zeta_{2}})({\zeta_{2}}-\zeta^{*}_{1})({\zeta_{1}}-\zeta^{*}_{1}),\\
&\Theta_{1}=\theta_{2}-\theta^{*}_{2},\quad \Theta'_{1}=-\theta_{1}-\theta^{*}_{1},\quad \Theta_{2}=\theta^{*}_{2}+\theta_{2},\quad \Theta'_{2}=\theta_{1}-\theta^{*}_{1}.\\
\end{split}
\end{equation}
Let $\zeta_{1}=\xi_{1}+\eta_{1}$,  $\zeta_{2}=\xi_{2}+\eta_{2}.$

Starting with a simple case, when the spectral parameter $\zeta_{1}$, $\zeta_{2}$  are pure imaginary numbers, i.e. $\xi_{1}=\xi_{2}=0$
\begin{equation}
Re(\theta_{1})=\eta_{1}(4\,t\alpha\,{\eta_{1}}^{2}+x),\\
Re(\theta_{2})=\eta_{2}(4\,t\alpha\,{\eta_{2}}^{2}+x).
\end{equation}
When $\alpha=0$,  the two constituent solitons have equal velocities, thus they will stay together and form a bound state. In a frame moving at this speed, this bound state will be spatially localized, and its amplitude function $|u(x, t)|$ will oscillate periodically with time. Let $\zeta_{1}= 0.7i, \zeta_{2}=0.4i, c_{1}=c_{2}=1$, such a bound state is illustrated in Figure 2 with \textbf{Case A} and \textbf{Case D}. It can be seen that the "width" of this solution changes periodically with time, thus this solution is called a "breather" in literature. When $\alpha\neq0$, two soliton do not form bound states, but the attraction ability between solitons will change with the change of parameters, see the figures of \textbf{Case A} and \textbf{Case D} at below.\\
At the follwing figures we notice:\\
\textbf{Case A:}  $\alpha=0$, $\gamma=1$. In this case the ENLS equation will be decayed to LPD equation, and $u(x,t)$ is the soliton solution of LPD equation. \\
\textbf{Case B:}  $\alpha=1$, $\gamma=1$. This is the ENLS equation.\\
\textbf{Case C:}  $\alpha=1$, $\gamma=0$. In this case the ENLS equation will be decayed to Hirota equation, and $u(x,t)$ is the soliton solution of the Hirota equation.\\
\textbf{Case D:}  $\alpha=0$, $\gamma=0$. In this case the ENLS equation will be decayed to NLS equation, and $u(x,t)$ is the soliton solution of the NLS equation.\\
\begin{figure}
\centering
\includegraphics[width=4cm,height=4cm]{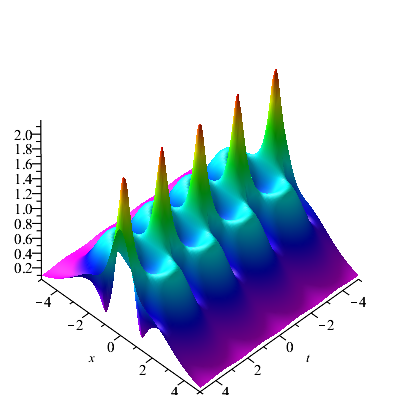}
$A$
\includegraphics[width=4cm,height=4cm]{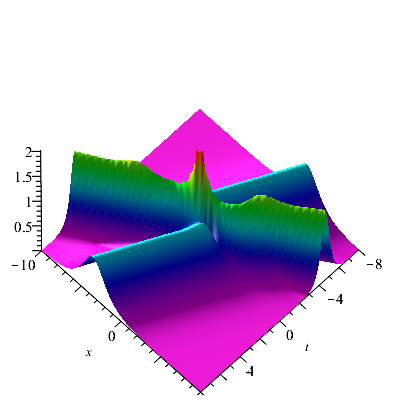}
$B$\\
\includegraphics[width=4cm,height=4cm]{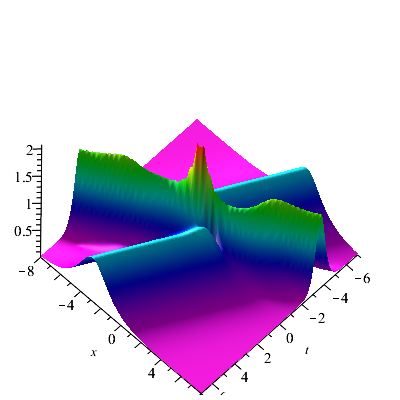}
$C$
\includegraphics[width=4cm,height=4cm]{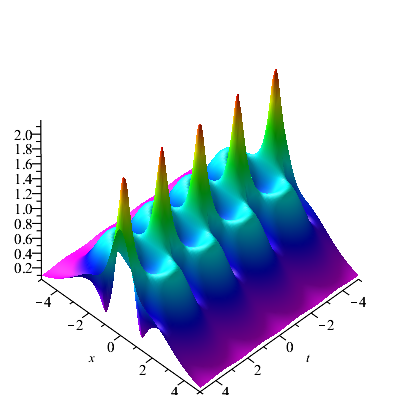}
$D$
\caption{3-D plot for the 2-soliton solution evolution of case A, B, C and D. $\zeta1=0.7i, \zeta2=0.4i.$}
\label{t2}
\end{figure}
\begin{figure}
\centering
\includegraphics[width=4.5cm,height=4cm]{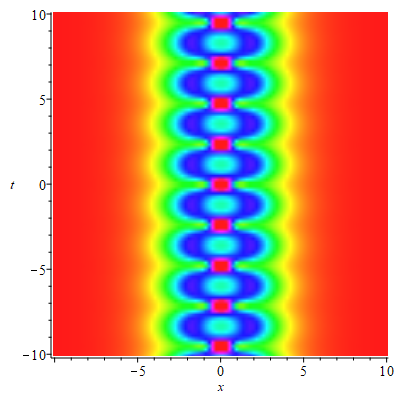}
$A$
\includegraphics[width=4.5cm,height=4cm]{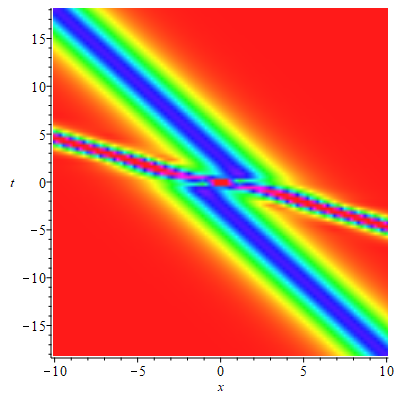}
$B$\\
\includegraphics[width=4.5cm,height=4cm]{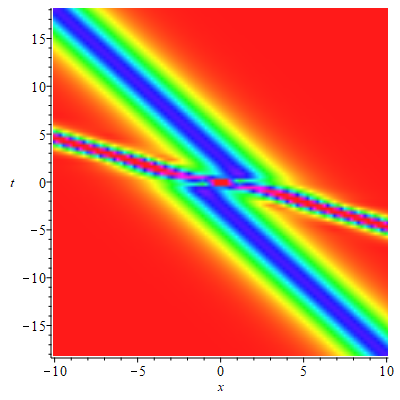}
$C$
\includegraphics[width=4.5cm,height=4cm]{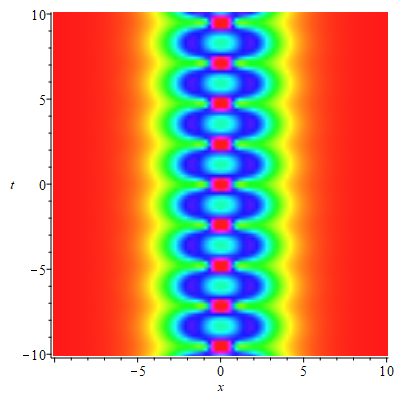}
$D$
\caption{The  plot for the 2-soliton solution evolution of case A, B, C and D.  $\zeta1=0.7i, \zeta2=0.4i.$}
\label{t3}
\end{figure}

Except for the case where both spectral parameters are purely imaginary, let's consider the more complex case. \\
Let $\zeta_{1}= 0.1+0.7i, \zeta_{2}=-0.1+0.4i, c_{1}=c_{2}=1$. We see from Figure 4 that as $t \rightarrow-\infty,$ the solution consists of two single solitons which are far apart and moving toward each other. When they collide, they interact strongly. But when $t \rightarrow \infty$, these solitons re-emerge out of interactions without any change of shape and velocity, and there is no energy radiation emitted to the far field.
Thus the interaction of these solitons is elastic (see Figure 6). This elastic interaction is a remarkable property which signals that the ENLS equation is integrable. There is still some trace of the interaction however. Indeed, after the interaction, each soliton acquires a position shift and a phase shift (see Figure 5). The position of each soliton is always shifted forward (toward the direction of propagation), as if the soliton accelerates during interactions.

Figure 4 is typical of all two-soliton solutions \eqref{2.1241} except $\xi_{1} = \xi_{2}=0$, we can easy found that $(\alpha,\gamma)$ will change the velocity, phase of the soliton figure. We analyze the asymptotic states of the solution \eqref{2.1241} as $t \rightarrow \pm \infty$ and  $(\alpha,\gamma)$ is non-negative.  Without loss of generality, Let $\zeta_{k}=\xi_{k}+i\eta_{k}$ and assume that  $|\xi_{1}|>|\xi_{2}|$. This means that at $t=-\infty,$ soliton-1 is on the right side of soliton-2 and moves slower. Note also that $\eta_{k}>0$ and $\eta_{2}>\eta_{1}$, since $\zeta_{k} \in \mathbb{C}_{+} .$ In the moving frame with velocity -($32\,{\eta_{1}}^{2}\gamma\,\xi_{1}-32\,\gamma\,{\xi_{1}}^{3}+4\,\alpha\,{\eta_{1}}^
{2}-12\,\alpha\,{\xi_{1}}^{2}+2\,\xi_{1}$), so $(\alpha,\gamma)$ will influence the velocity.
\begin{equation}
\begin{split}
&Re(\theta_{1})=\eta_{1} \left( 32\,t\gamma\,\xi_{1}\,{\eta_{1}}^{2}-32\,t\gamma\,{\xi_{1}}^{3}+4
\,t\alpha\,{\eta_{1}}^{2}-12\,t\alpha\,{\xi_{1}}^{2}+2\,t\xi_{1}+x \right)=O(1),\\
&Re(\theta_{2})=\eta_{2} \left( 32\,t\gamma\,\xi_{2}\,{\eta_{1}}^{2}-32\,t\gamma\,{\xi_{1}}^{3}+4
\,t\alpha\,{\eta_{1}}^{2}-12\,t\alpha\,{\xi_{1}}^{2}+2\,t\xi_{1}+x \right)\\
&+\eta_{2}((32\,t\gamma\,\xi_{2}\,{\eta_{2}}^{2}-32\,t\gamma\,\xi_{2}\,{\eta_{1}}^{2})+(32\,\gamma\,({\xi_{1}}^{3}-{\xi_{2}}^{3})+4
\,\alpha\,({\eta_{2}}^{2}-{\eta_{1}}^{2})\\&
+12\,\alpha\,({\xi_{1}}^{2}-{\xi_{2}}^{2})+2\,(\xi_{2}-2\,\xi_{1})))t.
\end{split}
\end{equation}
When $t \rightarrow-\infty, \operatorname{Re}\left(\theta_{2}\right) \rightarrow+\infty .$
When $t \rightarrow+\infty, \operatorname{Re}\left(\theta_{2}\right) \rightarrow-\infty .$
In this case, simple calculations show that the asymptotic state of the solution \eqref{2.124} is
\begin{equation}
u(x, t) \rightarrow
\begin{cases}
2 i\left(\zeta_{1}^{*}-\zeta_{1}\right) \frac{c_{1}^{-} e^{\theta_{1}-\theta_{1}^{*}}}{e^{-\left(\theta_{1}+\theta_{1}^{*}\right)}+\left|c_{1}^{-}\right|^{2} e^{\theta_{1}+\theta_{1}^{*}}}, \quad t \rightarrow-\infty,\\
\\
2 i\left(\zeta_{1}^{*}-\zeta_{1}\right) \frac{c_{1}^{+} e^{\theta_{1}-\theta_{1}^{*}}}{e^{-\left(\theta_{1}+\theta_{1}^{*}\right)}+\left|c_{1}^{+}\right|^{2} e^{\theta_{1}+\theta_{1}^{*}}}, \quad t \rightarrow+\infty,
\end{cases}
\end{equation}
where $c_{1}^{-}= \frac{c_{1}\left(\zeta_{1}-\zeta_{2}\right)}{\left(\zeta_{1}-\zeta_{2}^{*}\right)}$, $c_{1}^{+}=\frac{c_{1}\left(\zeta_{1}-\zeta_{2}^{*}\right)}{\left(\zeta_{1}-\zeta_{2}\right)}$. Comparing this expression with \eqref{2.1241}, we see that this asymptotic solution is a single-soliton solution with peak amplitude $2\eta_{1}$ and velocity $32\,{\eta_{1}}^{2}\gamma\,\xi_{1}-32\,\gamma\,{\xi_{1}}^{3}+4\,\alpha\,{\eta_{1}}^
{2}-12\,\alpha\,{\xi_{1}}^{2}+2\,\xi_{1}$.
 This indicates that if we fix the parameters this soliton does not change its shape and velocity after collision.

\begin{figure}
\centering
\includegraphics[width=4cm,height=4cm]{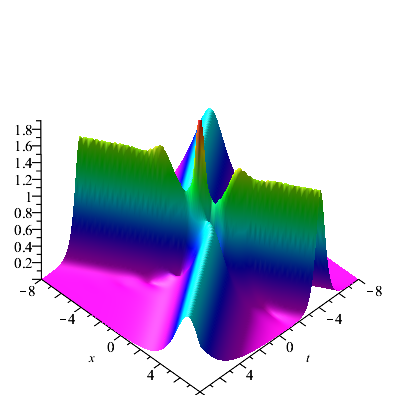}
$A$
\includegraphics[width=4cm,height=4cm]{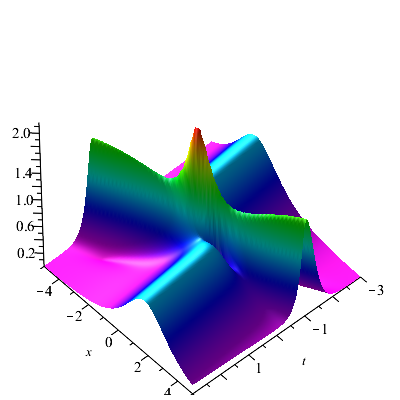}
$B$\\
\includegraphics[width=4cm,height=4cm]{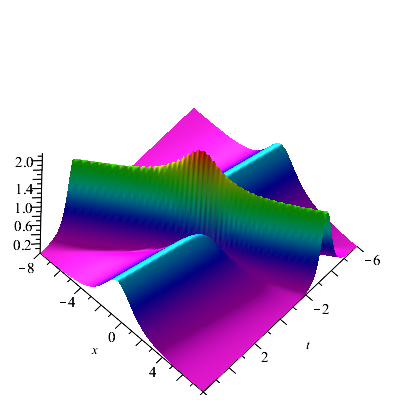}
$C$
\includegraphics[width=4cm,height=4cm]{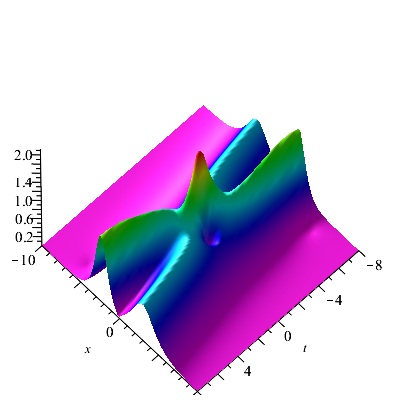}
$D$
\caption{3-D plot for the 2-soliton solution evolution of case A, B, C and D. $\zeta1=0.1+0.7i, \zeta2=-0.1+0.4i.$}
\label{t4}
\end{figure}
\begin{figure}
\centering
\includegraphics[width=4cm,height=4cm]{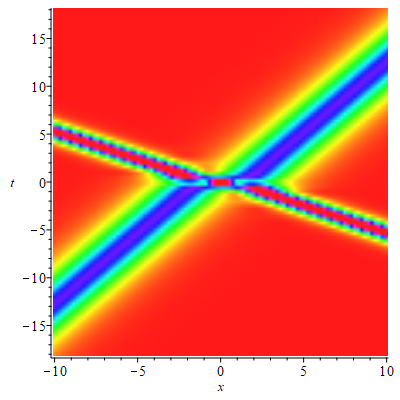}
$A$
\includegraphics[width=4cm,height=4cm]{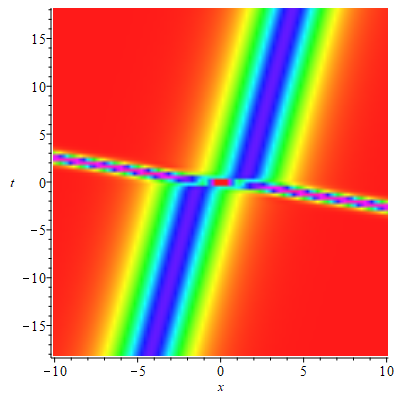}
$B$\\
\includegraphics[width=4cm,height=4cm]{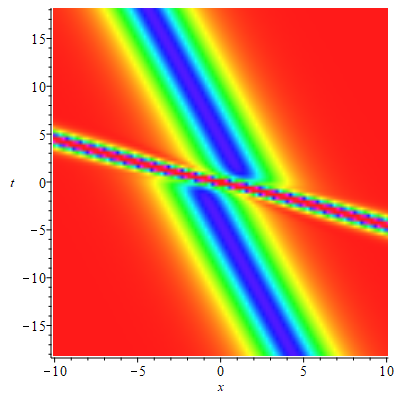}
$C$
\includegraphics[width=4cm,height=4cm]{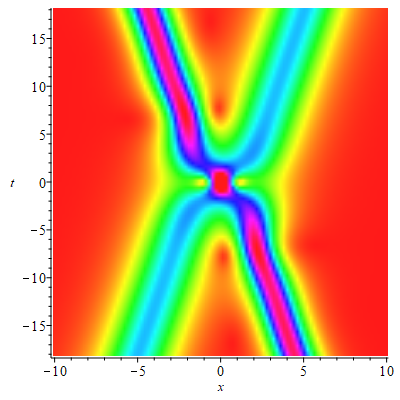}
$D$
\caption{The  plot for the 2-soliton solution evolution of case A, B, C and D.  $\zeta1=0.1+0.7i, \zeta2=-0.1+0.4i.$}
\label{t5}
\end{figure}
\begin{figure}
\centering
\includegraphics[width=4.5cm,height=4.2cm]{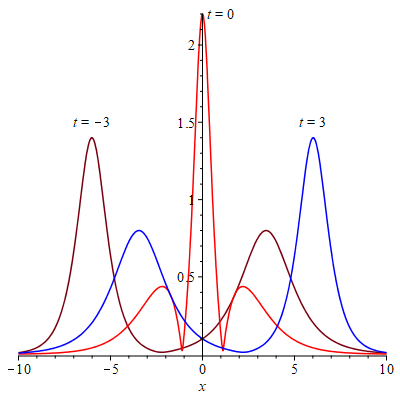}
\centering
\includegraphics[width=4.5cm,height=4.2cm]{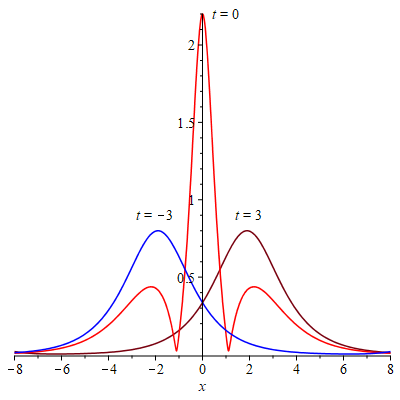}\\
\centering
\includegraphics[width=4.5cm,height=4.2cm]{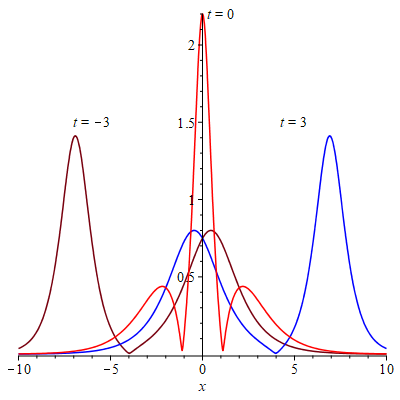}
\centering
\includegraphics[width=4.5cm,height=4.2cm]{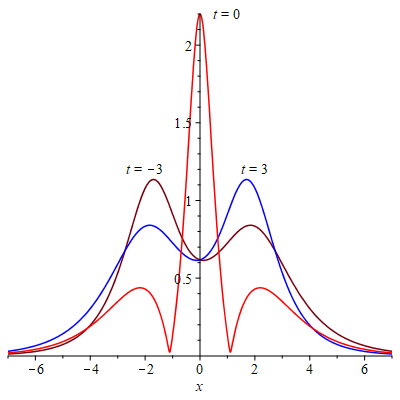}
\caption{The  plot for the 2-soliton solution evolution of case A, B, C and D. $\zeta1=0.1+0.7i, \zeta2=-0.1+0.4i.$}
\label{t6}
\end{figure}

\section{Soliton matrix for high-order zeros}

In this case, following the discussion of simple zeros, we consider the high-order zeros in RHP of ENLS equation. First of all, we let functions $P^{+}(\zeta)$ and $P^{-}(\zeta)$ from above RHP have only one pair of zero of order n, i.e. ${\zeta_{1} ,\bar{\zeta}_{1}}$.
\begin{equation}
|P^{+}(\zeta)|= (\zeta-\zeta_{1})^{n} \varphi(\zeta), \ \ \
|P^{-}(\zeta)|=(\zeta-\bar{\zeta}_{1})^{n}\bar{\varphi}(\zeta),
\end{equation}
where $|\varphi(\zeta_{ 1} )|\neq 0 $ and $|\bar{\varphi}(\bar{\zeta}_{ 1} )|\neq 0$.

Following the idea proposed in \cite{yjk2003a}, we first consider the elementary zero case under the assumption that the geometric multiplicity of $\zeta_{1}$ and $\bar{\zeta}_{1}$ has the same number. Hence, one needs to construct the dressing matrix $\Gamma(\zeta)$ whose determinant is $\frac{(\zeta-\zeta_{1})^{n}}{(\zeta-\bar{\zeta}_{1})^{n}}$. As a special case, we first consider the elementary zeros which have geometric multiplicity 1. In this case, $\Gamma$ is constituted of $n$ elementary dressing factors, i.e. : $\Gamma=\chi_{n}\chi_{n-1} \ldots \chi_{1}$,  where
\[
\begin{array}{l}
\chi_{i}(\zeta)=\mathbb{I}+\frac{\bar{\zeta}_{1}-\zeta_{1}}{\zeta-\bar{\zeta}_{1}} \mathbb{P}_{i}, \quad \mathbb{P}_{i}=\frac{|v_{i}\rangle\langle\bar{v}_{i}|}{\langle\bar{v}_{i} | v_{i}\rangle},\quad |v_{i}\rangle \in \operatorname{Ker}(P_{+} \chi_{1}^{-1} \cdots \chi_{i-1}^{-1}(\zeta_{1})). \\
\end{array}
\]
In addition, if we let $\hat{P}^{+}(\zeta)=P^{+}(\zeta) \chi_{1}^{-1}(\zeta)$ and $\hat{P}^{-}(\zeta)=\chi_{1}(\zeta) P^{-}(\zeta),$ then it is proved that matrices $\hat{P}^{+}(\zeta)$ and $\hat{P}^{-}(\zeta)$ are still holomorphic in the respective half plans of $\mathbb{C}$. Moreover, $\zeta_{1}$ and $\bar{\zeta}_{1}$ are still a pair of zeros of $|\hat{P}^{+}(\zeta)|$ and $|\hat{P}^{-}(\zeta)|$, respectively. Thus, $\Gamma(\zeta)^{-1}$ cancels all the high-order zeros for $|P^{+}(\zeta)|$. Moreover, it is necessary to reformulate the dressing factor into summation of fractions, then we derive the soliton matrix $\Gamma(\zeta)$ and its inverse for a pair of an elementary high-order zero. The results can be formulated in the following lemma.

\begin{lemma}
 Consider a pair of an elementary high-order zero of order $n:\left\{\zeta_{1}\right\}$ in $\mathbb{C}_{+}$ and $\left\{\bar{\zeta}_{1}\right\}$ in
$\mathbb{C}_{-}$. Then the corresponding soliton matrix and its inverse can be cast in the following form
\begin{equation}
\label{s51}
\Gamma^{-1}(\zeta)=I+\left(\left|p_{1}\right\rangle, \cdots,\left|p_{n}\right\rangle\right) \mathcal{D}(\zeta)\left(\begin{array}{c}
\left\langle q_{n}\right| \\
\vdots \\
\left\langle q_{1}\right|
\end{array}\right), \quad
\Gamma(\zeta)=I+\left(\left|\bar{q}_{n}\right\rangle, \cdots,\left|\bar{q}_{1}\right\rangle\right) \bar{\mathcal{D}}(\zeta)\left(\begin{array}{c}
\left\langle\bar{p}_{1}\right| \\
\vdots \\
\left\langle\bar{p}_{n}\right|
\end{array}\right),
\end{equation}
where $\mathcal{D}(\zeta)$ and $\bar{\mathcal{D}}(\zeta)$ are $n \times n$ block matrices,

\[
\begin{array}{c}
\mathcal{D}(\zeta)=\left(\begin{array}{ccccc}
(\zeta-\zeta_{1})^{-1} & (\zeta-\zeta_{1})^{-2} & \cdots & (\zeta-\zeta_{1})^{-n} \\
0 & \ddots & \ddots & \vdots \\
\vdots & \ddots & (\zeta-\zeta_{1})^{-1} & (\zeta-\zeta_{1})^{-2} \\
0 & \cdots & 0 & (\zeta-\zeta_{1})^{-1}
\end{array}\right), \quad  \\ \bar{\mathcal{D}}(\zeta)=\left(\begin{array}{cccc}
(\zeta-\zeta_{1})^{-1} & 0 & \cdots & 0 \\
(\zeta-\zeta_{1})^{-2} & (\zeta-\zeta_{1})^{-1} & \ddots & \vdots \\
\vdots & \ddots & \ddots & 0 \\
(\zeta-\zeta_{1})^{-n} & \cdots & (\zeta-\zeta_{1})^{-2} & (\zeta-\zeta_{1})^{-1}
\end{array}\right).
\end{array}
\]
\end{lemma}
This lemma can be proved by induction as in \cite{yjk2003a}. Besides, we notice that in the expressions for $\Gamma^{-1}(\zeta)$ and $\Gamma(\zeta),$ only half of the vector parameters, i.e.: $\left|p_{1}\right\rangle, \cdots,\left|p_{n}\right\rangle$
 and  $\left\langle \bar{p}_{1}\right|,  \cdots,\left\langle \bar{p}_{n}\right|$
    are independent. In fact, the rest of the vector parameters in \eqref{s51} can be derived by calculating the poles of each order in the identity $\Gamma(\zeta) \Gamma^{-1}(\zeta)=I$ at $\zeta=\zeta_{1}$
\[
\Gamma\left(\zeta_{1}\right)\left(\begin{array}{c}
\left|p_{1}\right\rangle \\
\vdots \\
\left|p_{n}\right\rangle
\end{array}\right)=0,
\]
where
\[
\Gamma(\zeta)=\left(\begin{array}{cccc}
\Gamma(\zeta) & 0 & \cdots & 0 \\
\frac{d}{d \zeta} \Gamma(\zeta) & \Gamma(\zeta) & \ddots & \vdots \\
\vdots & \ddots & \ddots & 0 \\
\frac{1}{(n-1) !} \frac{d^{n-1}}{d \zeta^{n-1}} \Gamma(\zeta) & \cdots & \frac{d}{d \zeta} \Gamma(\zeta) & \Gamma(\zeta)
\end{array}\right).
\]

Hence, in terms of the independent vector parameters, results \eqref{s51} can be formulated in a more compact form as in \cite{yjk2003a} and here we just avoid these overlapped parts. In the following, we derive this compact formula via the method of gDT \cite{llm2}. We intend to investigate the relation between dressing matrices and DT for ENLS equation in the high-order zero case. The essence of DT is a kind of gauge transformation. Following the project proposed in \cite{llm2015}, we can construct the gDT for ENLS equation as well.

The elementary form of DT has already been constructed in \cite{ybss2019}, then it is obvious to notice that:
$G_{1}\left(\zeta_{1}+\epsilon\right)\left|v_{1}\left(\zeta_{1}+\epsilon\right)\right\rangle=0 $. Denoting $\left|\chi_{1}^{[0]}\left(\zeta_{1}\right)\right\rangle=\left|v_{1}\left(\zeta_{1}\right)\right\rangle$,
and considering the following limitation:
\[
\left|\chi_{1}^{|1|}\left(\zeta_{1}\right)\right\rangle \triangleq \lim _{\epsilon \rightarrow 0} \frac{G_{1}\left(\zeta_{1}+\epsilon\right)\left|\chi_{1}^{|0|}\left(\zeta_{1}+\epsilon\right)\right\rangle}{\epsilon}=\frac{d}{d \zeta}\left[G_{1}(\zeta)\left|\chi_{1}^{[0]}(\zeta)\right\rangle\right]_{\zeta=\zeta_{1}},
\]
then $\left|\chi_{1}^{(1)}\right\rangle$ can be used to construct the next step DT, i.e.:
\[G_{1}^{[1]}(\zeta)=\left(I+\frac{\bar{\zeta}_{1}-\zeta_{1}}{\zeta-\bar{\zeta}_{1}} \mathbb{P}_{1}^{[1]}\right),\quad \mathbb{P}_{1}^{[1]}=\frac{\left|\chi_{1}^{[1]}\right\rangle\left\langle\chi_{1}^{[1]}\right|}{\left\langle\chi_{1}^{[1]} | \chi_{1}^{[1]}\right\rangle}.
\]
Generally, continuing this process we obtain:
\[
\left|\chi_{1}^{[N]}\right\rangle=\lim _{\epsilon \rightarrow 0} \frac{G_{1}^{[N-1]} \ldots G_{1}^{[1]} G_{1}^{[0]}\left(\zeta_{1}+\epsilon\right)\left|\chi_{1}^{[0]}\left(\zeta_{1}+\epsilon\right)\right\rangle}{\epsilon^{N}}.
\]

The N-times generalized Darboux matrix can be represented as:
\[
T_{N}(\zeta)=G_{1}^{[N-1]} \ldots G_{1}^{[1]} G_{1}^{[0]}(\zeta),
\]
where
\[G_{1}^{[i]}(\zeta)=\left(I+\frac{\bar{\zeta}_{i}-\zeta_{i}}{\zeta-\bar{\zeta}_{i}} \mathbb{P}_{1}^{[i]}\right),\quad \mathbb{P}_{1}^{[i]}=\frac{\left|\chi_{1}^{[i]}\right\rangle\left\langle\chi_{1}^{[i]}\right|}{\left\langle\chi_{1}^{[i]}| \chi_{1}^{[i]}\right\rangle}.
\]

In addition, the transformation between different potential matrices is:
\[
Q^{(N)}=Q+i\left[\Lambda, \sum_{j=0}^{N-1}\left(\bar{\zeta_{1}}-\zeta_{1}\right) \mathbb{P}_{1}^{[j]}\right].
\]
In this expression, $P_{1}^{[i]}$ is rank-one matrix, so $G_{1}^{[i]}(\zeta)$ can be also decomposed into the summation of simple fraction, that means the multiple product form of $T_{N}$ can be directly simplified by the conclusion of Lemma $1 .$ In other words, the above generalized Darboux matrix for ENLS equation can be given in the following theorem:
\\

\begin{theorem}
In the case of one pair of elementary high-order zero, the generalized Darboux matrix for ENLS equation can be represented as \cite{ybss2019}:
\[
T_{N}=I-Y M^{-1} \bar{\mathcal{D}}(\zeta) Y^{\dagger},
\]
where $\bar{\mathcal{D}}(\zeta)$ is $N \times N$ block Toeplitz matrix which has been given before, $Y$ is a $2 \times  N$ matrix:
\[
\begin{array}{c}
Y=\left(\left|v_{1}\right\rangle, \ldots, \frac{\left|v_{1}\right\rangle^{(N-1)}}{(N-1) !}\right), \\
\left|v_{1}\right\rangle^{(j)}=\lim _{\epsilon \rightarrow 0} \frac{d^{j}}{d \epsilon^{j}}\left|v_{1}\left(\zeta_{1}+\epsilon\right)\right\rangle,
\end{array}
\]
and $M$ is $N \times N$ matrix:
\[
M=\left(\begin{array}{ll}
M^{[ij]}
\end{array}\right), \quad M^{[i j]}=\left(M_{l, m}^{[\mathrm{i}, j]}\right)_{N \times N},
\]
with
\[
M_{l, m}^{[i, j]}=\lim _{\epsilon, \bar{\epsilon} \rightarrow 0} \frac{1}{(l-1) !(m-1) !} \frac{\partial^{m-1}}{\partial \epsilon^{m-1}} \frac{\partial^{l-1}}{\partial \bar{\epsilon}^{l-1}}\left[\frac{\left\langle y_{i} | y_{j}\right\rangle}{\zeta_{j}-\bar{\zeta}_{i}+\epsilon-\bar{\epsilon}}\right].
\]
\end{theorem}
Theorem 1 can be proved via directly calculation as in \cite{llm2015}.

Therefore, if $\Phi^{|N|}=T_{N} \Phi,$ then $\Phi^{[N]}$ indeed solves spectral problem \eqref{b}.
Substituting $T_{N}$ into the above relation and letting spectral $\zeta$ go to infinity, we have the relation:
\[
Q^{[N]}=Q-i\left[\Lambda,\left(\left|v_{1}\right\rangle, \ldots, \frac{\left|v_{1}\right\rangle^{(N-1)}}{(N-1) !}\right) M^{-1}\left(\begin{array}{c}
\left\langle v_{1}\right| \\
\vdots \\
\frac{\langle v_{1}|^{(N-1)}}{(N-1) !}
\end{array}\right)\right].
\]
Moreover, the transformations between the potential functions are:
\begin{equation}\label{59}
Q_{j, l}^{[N]}=Q_{j, l}^{[0]}+2i\frac{|A_{j, l}|}{|M|}, \quad A_{j, l}=\left[\begin{array}{cc}
M & Y[l]^{\dagger} \\
Y[j] & 0
\end{array}\right], 1 \leq j, l \leq 2.
\end{equation}
Here the subscript $_{j, l}$ denotes the $j$ th row and $l$ th column element of matrix $A$, and $Y[l]$ represents the $j$ th row of matrix $Y$.

\section{Dynamics of high-order solitons in the ENLS equation}

For simple, we consider the second-order fundamental soliton, which corresponds to a single pair of purely imaginary eigenvalues, $\zeta_{1}=$ $i \eta_{1} \in i \mathbb{R}_{+},$ and $\bar{\zeta}_{1}=i \bar{\eta}_{1} \in i \mathbb{R}_{-},$ where $\eta_{1}>$
0 and $\bar{\eta}_{1}=-\eta_{1}<0$. In this case, taking $v_{10}(\epsilon)=\left[1, \mathrm{e}^{i \theta_{10}-\theta_{11}\epsilon}\right]^{\mathrm{T}}$ and $\bar{v}_{10}(\bar{\epsilon})=$
$\left[1, \mathrm{e}^{i \bar{\theta}_{10}-\tilde{\theta}_{11} \bar{\epsilon}}\right]^{\mathrm{T}},$ where $\theta_{10}, \theta_{11}, \bar{\theta}_{10}, \bar{\theta}_{11}$ are real constants. Substituting these expressions into high-order soliton formula \eqref{59} with  $N=2, Q_{1,2}^{[0]}=0$. Then we obtain an analytic expression for the second-order fundamental soliton solution of \eqref{enls}
\begin{equation}\label{b29}
\begin{aligned}
&u(x, t)=\\
&2(\eta_{1}-\bar{\eta}_{1})\frac{t_{11}e^{2\eta_{1}x+(2i\eta_{1}^{2}+8\alpha \eta_{1}^{3}+16i\gamma \eta_{1}^{4})t+i\bar{\theta}_{10}}
+t_{12}e^{2\bar{\eta}_{1}x+(2i\bar{\eta}_{1}^{2}+8\alpha \bar{\eta}_{1}^{3}+16i\gamma \bar{\eta}_{1}^{4})t-i\theta_{10}}}{4 \cosh ^{2}(w)+F},
\end{aligned}
\end{equation}
\begin{equation}\notag
w=(\eta_{1}-\bar{\eta}_{1})x+[i(\bar{\eta}_{1}^{2}-\eta_{1}^{2})
+4\alpha(\bar{\eta}_{1}^{3}-\eta_{1}^{3})+8i\gamma(\bar{\eta}_{1}^{4}-\eta_{1}^{4})]t-\frac{i}{2}(\theta_{10}+\bar{\theta}_{10}),
\end{equation}
\begin{equation}\notag
t_{11}=(\bar{\eta}_{1}-\eta_{1})((64\bar{\eta}_{1}^{3}\gamma-24i\bar{\eta}_{1}^{2}\alpha+4\bar{\eta}_{1})t-2ix)-2i,
\end{equation}
\begin{equation}\notag
t_{12}=(\eta_{1}-\bar{\eta}_{1})((64\eta_{1}^{3}\gamma-24i\eta_{1}^{2}\alpha+4\eta_{1})t-2ix)-2i,\\
\end{equation}
\begin{equation}\notag
F(x,t)=(t_{11}+2i)(t_{12}+2i).
\end{equation}

Let $\bar{\eta}_{1}=-\eta_{1}, \bar{\theta}_{10}=-\theta_{10}$, $u(x,t)$ can be written in the form of a traveling solitary wave:
\begin{equation}\label{b29}
\begin{aligned}
&u(x,t)=\psi(x,t)e^{i(2\eta_{1}^{2}+16\gamma \eta_{1}^{4}-\theta_{10})},
\end{aligned}
\end{equation}
\begin{equation}
\begin{aligned}
&\psi(x,t)=4\eta_{1}\frac{(t_{11}e^{2\eta_{1}x+8\alpha \eta_{1}^{3}t}
+t_{12}e^{-2\eta_{1}x-8\alpha \eta_{1}^{3}t})}{4 \cosh ^{2}(2\eta_{1}x
-8\alpha\eta_{1}^{3}t)+F},
\end{aligned}
\end{equation}
\begin{equation}
\begin{aligned}
&|\psi(x,t)|^{2}=16\eta_{1}^{2}\frac{(t_{11}^{2}e^{4\eta_{1}x+16\alpha \eta_{1}^{3}t}
+t_{12}^{2}e^{-4\eta_{1}x-16\alpha \eta_{1}^{3}t}+2t_{11}t_{12})}{(4 \cosh ^{2}(2\eta_{1}x
-8\alpha\eta_{1}^{3}t)+F)^{2}}.
\end{aligned}
\end{equation}
The center trajectory $\Sigma_{+}$ and $\Sigma_{-}$ for this solution can be approximately described by the following two curves
\[\begin{aligned}
\Sigma_{+}: (\eta_{1}-\bar{\eta}_{1})x+4\alpha(\bar{\eta}_{1}^{3}-\eta_{1}^{3})t+\frac{1}{2}ln|F|=0,\\
\Sigma_{-}: (\eta_{1}-\bar{\eta}_{1})x+4\alpha(\bar{\eta}_{1}^{3}-\eta_{1}^{3})t-\frac{1}{2}ln|F|=0.
\end{aligned}\]
Moreover, regardless of the effect brought by the logarithmic part when $t \rightarrow \pm\infty$, two solitons separately move along each curve at nearly the same velocity, which is approximate to $4\alpha (\bar{\eta}_{1}^{2}+\eta_{1}\bar{\eta}_{1}+\eta_{1}^{2})$.

Due to $\eta_{1}-\bar{\eta}_{1}>0$, with simple calculation, it is found that $|u(x,t)|$ possesses the following asymptotic estimation:
\[\begin{aligned}
|u(x,t)|\rightarrow 0, \quad x \rightarrow \pm \infty.
\end{aligned}\]
However, with the development of time, a simple asymptotic analysis with estimation on the leading-order terms shows that: when soliton \eqref{b29} is moving on $\Sigma_{+}$ or $\Sigma_{-}$, its amplitudes $|u|$ can approximately vary as
\begin{equation}
u(x, t) \sim
\begin{cases}
\frac{2\left|\eta_{1}-\bar{\eta}_{1}\right| \mathrm{e}^{\left(\eta_{1}+\bar{\eta}_{1}\right)x}}{\left|\mathrm{e}^{(4i(\bar{\eta}_{1}^{2}-\eta_{1}^{2})+32i\gamma(\bar{\eta}_{1}^{4}-\eta_{1}^{4}))t-i (\operatorname{arg}[\mathcal{F}(x, t)]+2k\pi)+ i\left(\theta_{10}+\bar{\theta}_{10}\right)}+1\right|}, \quad t \sim  +\infty,\\
\\
\frac{2\left|\eta_{1}-\bar{\eta}_{1}\right| \mathrm{e}^{-(\eta_{1}+\bar{\eta}_{1})x}}{\left|\mathrm{e}^{-(4i(\bar{\eta}_{1}^{2}-\eta_{1}^{2})+32i\gamma(\bar{\eta}_{1}^{4}-\eta_{1}^{4}))t-i (\operatorname{arg}[\mathcal{F}(x, t)]+2k\pi)- i\left(\theta_{10}+\bar{\theta}_{10}\right)}+1\right|}, \quad t \sim -\infty,
\end{cases}
\end{equation}
$k\in \mathbb{Z}$.\\
Let $\eta_{1}= \frac{i}{2}, \bar{\eta}_{1}=-\frac{i}{2}, \theta_{10}=\bar{\theta}_{10}=\theta_{11}=\bar{\theta}_{11}=0$.
Because the effect brought by the logarithmic part, choose different parameter values $(\alpha, \gamma)$, two solitons separately move along each curve with different velocity, direction and shape (see Figure 7).
\begin{figure}
\centering
\includegraphics[width=4cm,height=4cm]{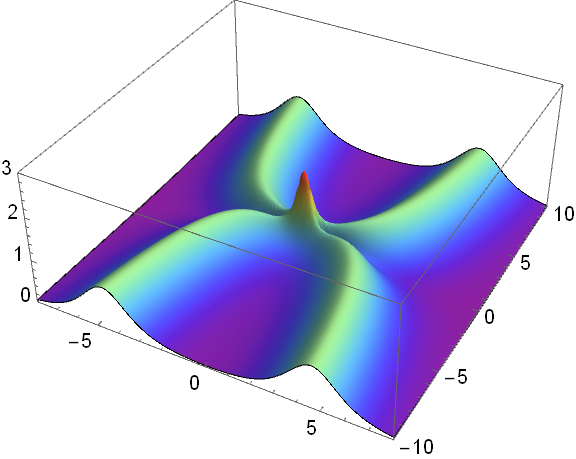}
$A$\quad
\includegraphics[width=4cm,height=4cm]{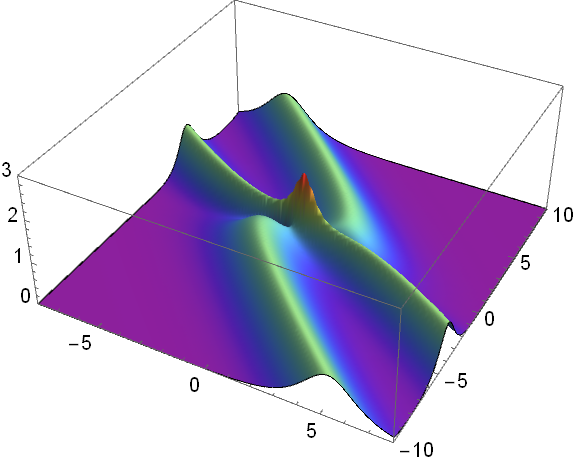}
$B$\\
\includegraphics[width=4cm,height=4cm]{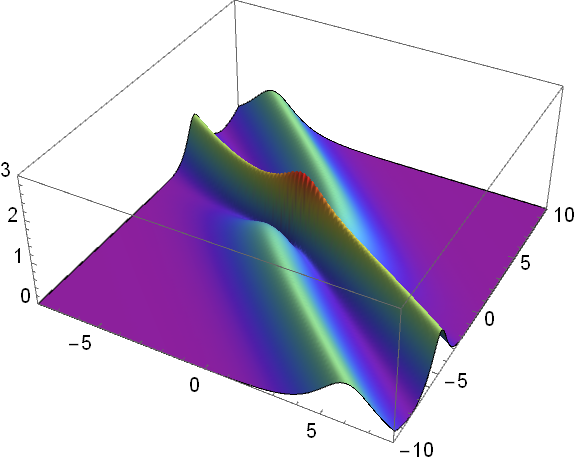}
$C$\quad
\includegraphics[width=4cm,height=4cm]{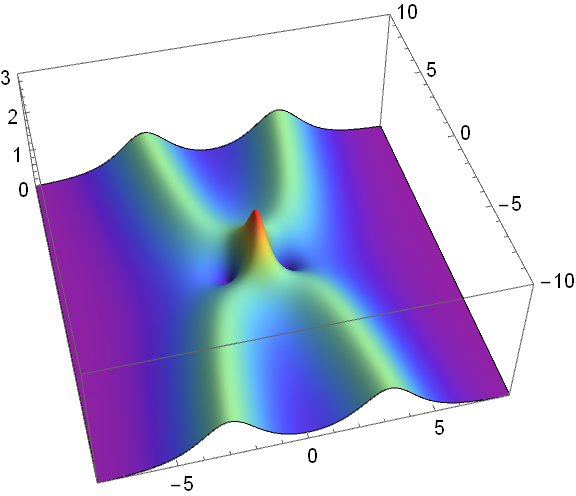}
$D$
\caption{3-D plot for the high-order solution evolution of case A, B, C and D.}
\label{t7}
\end{figure}

\section*{Acknowledgements}
 The project is supported by the National Natural Science Foundation of China (No.11675054), Shanghai Collaborative Innovation Center of Trustworthy Software for Internet of Things (No.ZF1213), and Science and Technology Commission of Shanghai Municipality (No.18dz2271000).

\end{document}